\documentclass[11pt]{article}
\usepackage{amsmath,amssymb,amsfonts,amsthm}
\numberwithin{equation}{section}
\usepackage[margin=1.15in]{geometry}
\usepackage[numbers]{natbib}
\usepackage[many]{tcolorbox}
\usepackage{enumitem}
\usepackage{subcaption}
\usepackage{xcolor}
\usepackage{booktabs}
\usepackage{array}
\usepackage{placeins}
\definecolor{darkgreen}{RGB}{0, 100, 0}
\usepackage[colorlinks=true,citecolor=darkgreen,linkcolor=blue]{hyperref}

\makeatletter
\newcommand{\nameditem}[1]{\item[(#1)]\def\@currentlabel{#1}\label}

\@tfor\cal@letter:=ABCDEFGHIJKLMNOPQRSTUVWXYZ\do{%
    \expandafter\edef\csname cal\cal@letter\endcsname{\noexpand\mathcal{\cal@letter}}%
}
\makeatother

\newcommand{\ind}[1]{\mathbf{1}\left\{#1\right\}}
\newcommand{\E}{\mathbb{E}}
\newcommand{\risk}{\mathbf{R}}
\renewcommand{\Re}{\mathbb{R}}
\newcommand{\w}{\mathrm{w}}
\renewcommand{\P}{\mathsf{P}} 
\newcommand{\Q}{\mathsf{Q}} 

\newcommand{\Mbh}{\calM^{\textnormal{\textsc{e-bh}}}}
\newcommand{\Mall}{\calM^{\textnormal{\textsc{AoN}}}}
\newcommand{\Mcum}{\calM^{\textnormal{\textsc{cum}}}}
\newcommand{\Monline}{\calM^{\textnormal{\textsc{on}}}}

\newcommand{\nap}{\textnormal{\textsf{NAP}}}
\newcommand{\nar}{\textnormal{\textsf{NAR}}}

\newtheorem{proposition}{Proposition}
\newtheorem{theorem}{Theorem}

\newtheorem{remark}{Remark}
\theoremstyle{definition}

\newtheorem{example}{Example}

\title{Strategic Multi-Agent Resource Allocation \\ via Statistical Contracts}
\author{Ben Chugg\thanks{\footnotesize Carnegie Mellon University, \url{benchugg@cmu.edu}} \quad Michael I. Jordan\thanks{\footnotesize UC Berkeley, ÉNS, \& Inria} \quad Aaditya Ramdas\thanks{\footnotesize Stanford University}}
\date{\today}

\newif\iflong
\longtrue

\begin{document}

\maketitle

\begin{abstract}
We study a class of strategic hypothesis testing problems in which a principal aims to allocate resources among competing agents who possess data or data-generating mechanisms that may be relevant to the principal's goals. Each agent can pay a participation cost to opt into a \emph{statistical contract}, in which case they generate data from a private distribution and receive an allocation based on a predetermined (and publicly known) mechanism. Each agent is either ``null'' or ``non-null'' and the principal's goal is to design the contract to limit the allocation to the former while maximizing transactions with the latter.

Using e-values, we define a ``continuous-closure,'' a set of allocations which prevents null agents from profiting in expectation and provides post-hoc risk control for general loss functions under arbitrary dependence.
When the loss is the fraction of resources allocated to null agents, the region admits polynomial-time linear optimization. We also give explicit, interpretable closed-form approximations, and we analyze participation from a non-null agent's perspective, giving conditions under which entry into the contract is profitable. 
Finally, we extend the framework to support online arrival of the agents, in which case the principal has to make irrevocable on-the-fly decisions, and provide a method which gives time-uniform error-control in this setting.  

\end{abstract}

\section{Introduction}

We consider allocation problems in which a principal must distribute resources among competing agents.
Each agent is associated with a distribution $P\in\calR$, where $\calR$ is partitioned into ``nulls'' $\calP$ and ``non-nulls'' $\calQ$. Each agent's distribution is unknown to the principal (but may be known to the agent), whose goal is to distribute resources to the non-null agents while avoiding null agents. To motivate such problems, consider the following scenarios.

\begin{example}
\label{ex:pharma}
Large pharmaceutical companies (principals) often rely on independent startups (agents) for drug discovery and early-stage clinical trials~\citep{vandevyver2025pulse}. Startups approach the company if they have a promising result. They conduct a pilot study to demonstrate the promise of their drug,  and the company allocates revenue (e.g., for further experimentation) based on the result.
\end{example}

\begin{example}
\label{ex:digital-platform}
A digital platform (the principal) allocates user traffic among competing ads, app-store listings, or product variants. Advertisers, developers, or internal product teams (agents) propose candidate interventions and evaluate them through the platform's experimentation infrastructure. The platform then uses the resulting evidence to decide on traffic share. Real-world examples include app-store optimization~\citep{apple2026productpage,google2026customexperiment} and advertising experiments~\citep{meta2026splittesting,kohavi2013online}.
\end{example}

Our approach falls under the purview of \emph{statistical contract theory}~\citep{bates2022principal}. As in classical contract theory~\citep{laffont2002theory}, the goal of statistical contract theory is to solve problems involving \emph{information asymmetry}---where a principal interacts with an agent who has problem-relevant knowledge that the principal does not possess, where the principal accordingly needs to incentivize the agent to participate in the solution of the problem, and where the overall social welfare reflects both the principal's goals and the agent's goals. Statistical contracts focus in particular on problems of information asymmetry involving data analysis, where the principal and the agent have different levels of access to statistical quantities such as parameters, distributions, and data.

When designed well, a statistical contract should allow the principal to solve an inferential problem---such as controlling false positives in a hypothesis test---via a protocol whereby an agent can verify that it is in their interest to contribute their private knowledge in the interaction with the principal.  Previous work in statistical contract theory has focused on the core problem of hypothesis testing in the interaction between a principal and a single agent~\citep{bates2022principal,bates2023incentive,chen2025strategic}. The current paper expands the scope of the theory to consider principal-agent hypothesis testing in the more general setting of multiple agents.

Let us briefly sketch the details of a contract here, which is defined by an allocation mechanism $\calM$. Each agent $i$ with private distribution $P_i$ can pay a fixed, publicly known cost $c_i>0$ to enter the contract. Upon entry, data $X_i\sim P_i$ is drawn from the agent's distribution.\footnote{Perhaps the data is drawn by the agent themselves, by the agent supervised by the principal, or by a trusted third-party. We remain agnostic to these details.}
The principal gathers data from all agents who opted into the contract and feeds them into  $\calM$ which provides a  final allocation for each (e.g., a share of the total budget, a profit license, or how widely to disseminate an agent's product). See Figure~\ref{fig:contract}. We assume that the mechanism to decide on the allocation given the data is publicly known. The goal of this paper is to design mechanisms that yield controllable incentive structures.

The main challenge is in designing a contract which ensures that the allocation to null agents is ``small'' (as determined by some loss function), even when the principal does not know the underlying distribution of each agent or which agents are null. The interaction between the principal and the agents is thus \emph{strategic}: Agents may have private information concerning their own distribution $P_i$, and can decide to enter the contract based on $c_i$, $P_i$, and perhaps their beliefs about the other agents.

\begin{figure}
    \centering
    \includegraphics[width=0.75\textwidth]{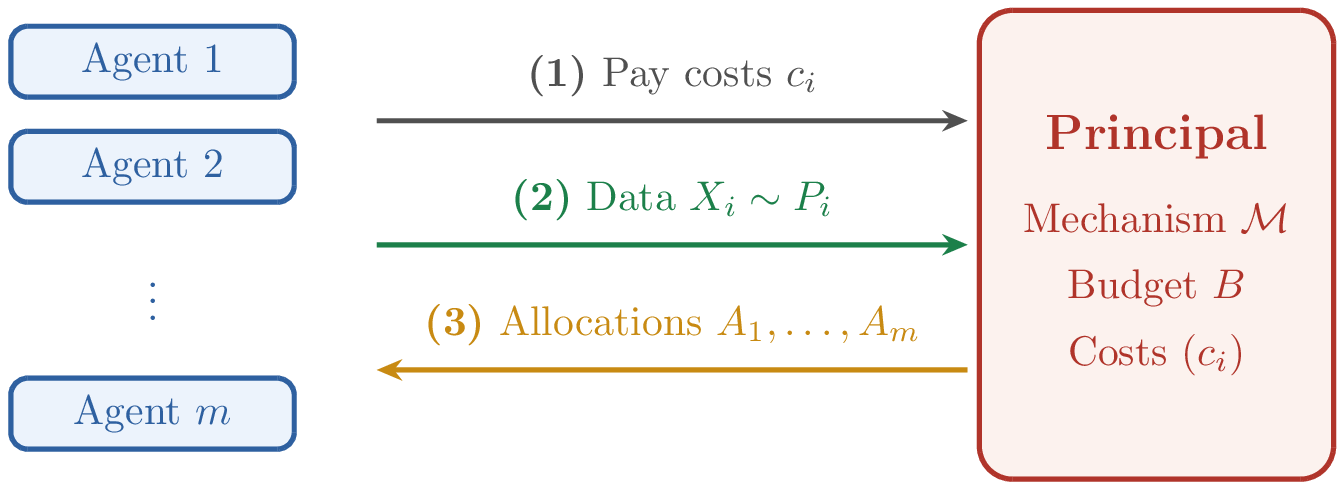}
    \caption{Strategic multi-agent resource allocation. The principal publicly announces the mechanism $\calM$, an optional budget $B$, and the costs $(c_i)$. Each agent can enter the contract, at which point they pay $c_i$ and generate data from their private distribution. The mechanism $\calM$ then computes an allocation for each agent.}
    \label{fig:contract}
\end{figure}

\paragraph{Contributions.}
After formalizing the multi-agent statistical contract in Section~\ref{sec:contract}, we introduce several desiderata that we would like a ``well-designed'' contract to satisfy. Following \citet{bates2022principal}, we say that a contract satisfies \emph{null deterrence} if the expected allocation to each null agent $i$ is at most their participation cost $c_i$. We then introduce the \emph{null allocation rate}---the expected fraction of the payout to null agents---and various other loss functions that the principal might like to control. We also introduce \emph{post-hoc risk control}, which is control of a loss function at data-dependent significance levels (thus strictly generalizing the notion of controlling the expected loss at a predetermined level) and simultaneously over several loss functions. From here, our main contributions are the following:
\begin{enumerate}
    \item In Section~\ref{sec:general-losses}, we introduce the \emph{continuous-closure} for arbitrary loss functions, a set of allocations which satisfy null deterrence and post-hoc risk control. The principal is able to freely maximize any objective function over this set.
    We then focus on the null allocation rate specifically, showing that one can optimize linear objectives over its continuous-closure in polynomial time. We then give three simple mechanisms which provide explicit, interpretable approximations to this optimum. We characterize the tradeoffs both analytically and experimentally.
    \item In Section~\ref{sec:nonnull-participation}, we turn to studying non-null agents, aiming to provide conditions under which it is beneficial to opt into the contract.
    We show that strong evidence alone need not make entry profitable, since competitors with faster evidence growth can absorb nearly all of the budget. However, a simple strategy of truncating the evidence yields explicit lower bounds for a non-null agent's payout. Moreover, we show that such truncation is necessary---in a particular formal sense---to provide such guarantees.
    \item Finally, Section~\ref{sec:extensions} considers sequential allocation. When agents arrive over time and allocations are irrevocable, we establish time-uniform control of the null allocation rate, while permitting the principal to allocate risk incrementally through a predictable policy. We also show how the allocation can fund another round of experimentation. Appendix~\ref{app:selective-entry} separately treats agents who observe their data before deciding whether to enter the contract.

\end{enumerate}
The allocations of all our mechanisms are based on \emph{e-values}, a statistical cousin of the p-value~\citep{ramdas2025hypothesis,grunwald2024safe}. Just as in the single-agent setting of \citet{bates2022principal}, e-values are necessary tools for designing contracts that satisfy null deterrence.
They are also very natural here: E-values are often considered as measures of evidence against the null hypothesis~\citep{chugg2026values}, designed to be small under $\calP$ and large under $\calQ$. Thus, allocations based on e-values will naturally favor agents in $\calQ$. One consequence of e-value-based allocations is that our batch mechanisms are robust under arbitrary dependence across agents, and can easily be extended to support online allocation.

\paragraph{Setting and preliminaries.}
As stated above, we consider a universe of distributions $\calR$ partitioned into null distributions $\calP$ and non-null distributions $\calQ$. Each agent, which we identify with a number $i\in\mathbb{N}$, is associated with a distribution $P_i\in \calR$ unknown to the principal (and perhaps not even fully known to the agent). Null agents have $P_i\in\calP$. An e-value for $\calP$ is a nonnegative random variable $E$ such that $\E_P[E]\leq 1$ for all $P\in\calP$.

Throughout, we enumerate the set of agents who opted into the contract as $1,\dots,m$, and we let $N\subseteq[m]$ be the set of null agents. Agent $i$ has a fixed, publicly known participation cost $c_i>0$. For each $S\subseteq[m]$, define $C_S=\sum_{i\in S}c_i$, and write $C=C_{[m]}$ for the total participation cost. We will always take the set $N$ of null agents to be unknown but fixed. Thus, $m$, $(c_i)_{i\in[m]}$, and $N$ are treated as constants in our analyses and results.
We write $\E_N[\cdot]$ as a shorthand for the worst-case expectation over joint distributions such that $N$ is the set of null agents.\footnote{More formally, let $\mathfrak P_N$ denote the class of joint laws
under which precisely the agents in $N$ are null, that is,
the marginal law $P_i$ of the data generated by agent $i$ belongs to $\calP$ for each $i\in N$ and to $\calQ$ for each $i\in[m]\setminus N$. Unless stated otherwise, \(\mathfrak P_N\) imposes no independence
assumptions among the agents' data. Then $\E_N[\cdot] = \sup_{R\in\mathfrak P_N} \E_R[\cdot]$ is the supremum of expectations over all such joint laws.}


\section{Related Work}
\label{sec:related-work}

This work sits at the intersection of multiple testing, contract theory, and resource allocation. It brings together questions of statistical validity, strategic participation, and optimization under resource constraints. Our view is that these connections create opportunities for the cross-fertilization of ideas in both directions.
We review the most relevant lines of work below.

\paragraph{Statistical contract theory.} The field of statistical contract theory was kickstarted by the original paper of \citet{bates2022principal} on principal-agent hypothesis testing, sometimes also called strategic hypothesis testing~\citep{chen2025strategic}. This work builds on the earlier analysis of \citet{tetenov2016economic}, who was concerned with  drug approval after phase III clinical trials and studied how to set the significance level to deter null agents.
There have been several follow-up works in this area, each providing a different  perspective~\citep{artola2026optimizing,bates2023incentive}.
A recent line of work studies the optimality properties of statistical contracts, including instance-specific and minimax characterizations of optimal tests~\citep{shi2024sharp,shi2025instance,chen2025strategic}.

\paragraph{Algorithmic contract theory.} There is a distinct literature on \emph{algorithmic} contract theory which studies how to design optimal contracts under a model of moral hazard~\citep{dutting2024algorithmic,zhu2022sample,saig2023delegated,saig2024incentivizing}.
There, an agent takes a  hidden, costly action that shapes the distribution of outcomes and the principal designs payments to maximize its expected utility. In our setting, by contrast, the agent's only action (whether to participate) is observed and does not change its data distribution. Still, algorithmic contract theory has been applied to multi-agent systems, so connections between the two subfields may be worth further investigation~\citep{dutting2023multi}.

\paragraph{Contract theory and mechanism design.}
Both statistical contract theory and algorithmic contract theory are computational perspectives on the more general area of contract theory, a rich area of economics~\citep{smith2004contract,bolton2004contract}. Standard questions in contract theory include how to design incentives under asymmetric information, how to induce agents to take costly actions, how to screen agents with heterogeneous private types, and how to structure rewards
when the principal observes only noisy signals of agent quality. Our setting shares these concerns, but differs in that the data are generated by formal statistical experiments and the principal seeks risk guarantees on the resulting allocation.

Likewise, principal-agent frameworks have also long been studied outside of statistics~\citep{grossman1992analysis}. Classical principal-agent models study incentives under asymmetric information, moral hazard, or hidden action and thus differ from the hypothesis testing perspective that we adopt here.

\paragraph{Multiple testing and e-values.}
Our methods have much in common with the multiple-testing literature. This is perhaps unsurprising, given that multiple testing can be seen as a form of binary resource allocation, where each hypothesis is an agent and we are making yes/no decisions.
Indeed, the continuous-closure introduced in Section~\ref{sec:general-losses} can be seen as a generalization of the e-closure principle~\citep{xu2025bringing}, and our e-BH mechanism in Section~\ref{sec:nar-control} draws inspiration from the e-BH procedure of \citet{wang2022false}. Our online mechanism is based on the e-LOND procedure of \citet{xu2024online}.
\citet{viviano2026model} recently studied multiple testing via an economic lens, asking how the error budget should be set as a function of research costs.

\paragraph{Resource allocation.} Many other frameworks are broadly focused on the question of resource allocation, including multi-armed and contextual bandits, portfolio optimization, auction design, and network optimization. Bandit algorithms allocate samples or traffic across alternatives while balancing exploration and exploitation~\citep{lattimore2020bandit,slivkins2019introduction}, though they do not typically take into account the strategic behavior on the part of the alternatives. Variants such as bandits with knapsacks~\citep{badanidiyuru2018bandits} take into account budget constraints, as we do here. Auction theory~\citep{menezes2005introduction} also studies allocation under uncertainty and strategic behavior, but does not impose statistical guarantees on the behavior of null agents.


\section{Multi-Agent Statistical Contracts}
\label{sec:contract}

Here we introduce the statistical contract which governs the interaction between the principal and the agents. We consider two scenarios: That of a finite budget, in which the total allocation across all agents must be bounded by a known constant $0<B<\infty$, and that of an infinite budget $B=\infty$. We emphasize that we use the term ``infinite budget'' for convenience only---it should not be taken literally as an infinite financial budget. Instead, it corresponds to situations in which there is no natural bound on the allocations produced by the principal. This arises, for example, when the principal is granting licenses to profit~\citep{bates2022principal,bates2023incentive}, market access, or exposure multipliers (e.g., in the case of ads in Example~\ref{ex:digital-platform}). In such settings, allocating to one agent does not reduce what can be allocated to another and yet, the principal would still like to optimize its allocation portfolio as a whole (not allocate very much to nulls).

\begin{tcolorbox}[
    colback=blue!8!white,
    colframe=white,
    boxrule=0pt,
    sharp corners,
    left=6pt,
    right=6pt,
    top=6pt,
    bottom=6pt,
    breakable
]
\textbf{Multi-agent statistical contract}
\begin{enumerate}[]
    \item The principal announces mechanism $\calM$, publicly known participation costs $(c_i)_{i\geq 1}$ with $c_i>0$, and budget $B\in(0,\infty]$.
    \item Each agent $i$ with private distribution $P_i$ decides whether to enter the contract. If they enter, they pay cost $c_i$ and \emph{then} collect their data $X_i\sim P_i$.
    \item The data $(X_1,\dots,X_m)$ from the $m$ agents that opted into the contract are fed into the mechanism $\calM$ which produces allocations $A_1,\dots,A_m\in\Re_{\geq 0}$. The allocations must satisfy $\sum_{i=1}^m A_i \leq B$.
\end{enumerate}
\end{tcolorbox}
Next we consider various desiderata that one might like the mechanism to satisfy.

\paragraph{Null deterrence.}
A first natural goal for the principal, following in the footsteps of \citet{bates2022principal}, is to construct the mechanism $\calM$ to de-incentivize null agents from participating. That is, the expected allocation for a null agent should not exceed the cost of entry. Formally, let us say that the mechanism $\calM$ satisfies \emph{null deterrence} (\ref{eq:ic}) if, for each null agent $i\in N$,
\begin{equation*}
\tag{ND}
\label{eq:ic}
    \E_{N}[A_i]\leq c_i.
\end{equation*}
Since allocations are always nonnegative, an easy consequence is that $\calM$ satisfies \ref{eq:ic} if and only if $A_i/c_i$ is an e-value for joint distributions such that $i\in N$.


\paragraph{Risk control.}
In addition to designing mechanisms satisfying \ref{eq:ic}, which is a constraint on each individual agent, the principal may be interested in controlling more general loss functions which take into account the portfolio of allocations.
Formally, a loss $\ell$ is an indexed family $\ell=(\ell_S)_{S\subseteq[m]}$ of nonnegative functions $\ell_S:\Re_{\geq0}^m\to\Re_{\geq0}$, where $\ell_S(A)$ is the loss at allocation $A=(A_1,\dots,A_m)$ when $S$ is the null set. Thus $\ell_N$ is the member of this family indexed by the true null set $N$. For simplicity, when the null set is clear from context, we may omit the subscript $N$ and
we will sometimes write $\ell_N(\calM)$
for the loss evaluated at the allocation produced by mechanism $\calM$.

Given a fixed loss function (we provide concrete examples below), we define two objectives. The first is \emph{risk control}, or expected loss control, at a prespecified level $\Delta>0$:
\begin{equation}
\label{eq:expected-loss}
    L_N(\calM_\Delta) := \E_N[\ell_N(\calM_\Delta)] \leq \Delta.
\end{equation}
We call $L_N(\calM_\Delta)$ the \emph{risk}. We have introduced the subscript on $\calM$ to indicate that the mechanism will take $\Delta$ as a parameter.

The second objective is more stringent. It draws inspiration from recent work on post-hoc inference~\citep{katsevich2020simultaneous,wang2022false,grunwald2024beyond,koning2023post,chugg2026admissibility}, the goal of which is to allow the error parameter to be data-dependent. Practically speaking, we would like to enable the principal to gather the data from the agents, compute the allocations with a prespecified error target $\Delta_0$ in mind, but then change its mind based on the result and compute the allocations with a different $\Delta_1\neq \Delta_0$. Such behavior is implicitly prohibited by~\eqref{eq:expected-loss}, which treats $\Delta$ as a data-independent constant.

This motivates the notion of \emph{post-hoc risk control}. Throughout, we assume that all allocation selections and all suprema appearing in risk bounds and certificates are measurable. Given a family of losses $\calL$ let us say that the family of mechanisms $\calM$ is \emph{post-hoc risk safe} with respect to $\calL$ if
\begin{equation}
\label{eq:expected-risk}
    \risk_N(\calM,\calL) := \E_N\left[\sup_{\ell \in\calL}\sup_{\Delta>0}\frac{\ell_N(\calM_\Delta)}{\Delta}\right]\leq 1.
\end{equation}
We call the quantity $\risk_N(\calM,\calL)$ the \emph{post-hoc risk}. Post-hoc risk control allows both the error bound $\Delta$ and the loss $\ell$ to be data-dependent, thus enabling the principal to change both and recompute the allocations as many times as desired.
It is perhaps worth emphasizing that in~\eqref{eq:expected-risk}, the set $N$ of null agents is fixed. The loss function itself may vary (and $\Delta$ alongside it), but $N$ is not a function of the data as each agent is either null or non-null before entering the contract. In other words, choosing a family $\ell\in\calL$ determines how $\ell_S$ behaves for all $S\subset[m]$, and $\ell_N$ denotes the member corresponding to the fixed null set $N$.

Observe that~\eqref{eq:expected-risk} is a stricter condition than~\eqref{eq:expected-loss}. Indeed, the former recovers the latter if we fix $\Delta$ and take $\calL=\{\ell\}$ to be a singleton. Many of the mechanisms we provide will be post-hoc risk safe; we will not continually belabor the point that this also means that they control $L_N(\calM_\Delta)$. Table~\ref{tab:loss-functions} in Appendix~\ref{app:loss-examples} lists several loss functions that would be sensible to consider. We turn to a particularly relevant loss next.

\paragraph{Null allocation rate.}
Section~\ref{sec:general-losses} presents allocation regions whose selections yield post-hoc risk-safe families of mechanisms for arbitrary losses. Thereafter, however, we will primarily be concerned with a specific loss---the \emph{null allocation proportion} (NAP). The motivation is that the principal might be concerned with the share of the total allocation going to null agents. We thus define the NAP, and its expectation, the \emph{null allocation rate} (NAR):
\begin{equation}
\begin{aligned}
    \nap_N(\calM)
    &:= \frac{\sum_{i\in N} A_i}{\sum_{i\in[m]} A_i},
    \qquad \nar_N(\calM)
    := \E_N[\nap_N(\calM)].
\end{aligned}
\end{equation}
We use the convention $0/0=0$ so that $\nap$ and $\nar$ are always well-defined and in $[0,1]$. We will often drop the subscript $N$ when it is clear from context.
If $\nar_N(\calM)\leq \Delta$ for any $\Delta>0$ we say that $\calM$ controls the NAR at level $\Delta$, or more generally that $\calM$ exhibits NAR control.

The reader familiar with multiple testing will recognize parallels to false discovery rate (FDR) control. In particular, if we enforce $A_i \in \{0,1\}$, then NAR control is equivalent to FDR control, albeit with the strategic entry of hypotheses (i.e., agents). Indeed, let $R = \{ i: A_i=1\}$,  in which case $\nap_N(\calM) = |N\cap R| / |R| $ is the false discovery proportion. This should give some hint that multiple testing is a particular instance of our framework.

Note that neither null deterrence nor NAR control implies the other. If one allocates, say, $c_i+1$ to each agent $i$, then the mechanism does not satisfy \ref{eq:ic}, but the NAR can be arbitrarily small when $(C_N+|N|)/(C+m)$ is small. Conversely, allocating $c_i$ to each agent satisfies \ref{eq:ic} but yields NAP $C_N/C$, which can be large. \\

Beyond these guarantees, the principal may wish to maximize an objective reflecting the value of the allocations (e.g., perhaps $V(A) = \sum_{i=1}^m v_i A_i$ or $V(A) = \sum_{i\in [m]}v_i\log(1 + A_i)$ for weights $v_i$; see Section~\ref{sec:computational-comparison}). Section~\ref{sec:general-losses} therefore constructs allocation regions over which the principal can optimize while preserving null deterrence and post-hoc risk control, while  Section~\ref{sec:nap} studies when this optimization is computationally tractable.

In addition to the principal's utility, we might also consider an agent's utility. In particular, a well-constructed contract should incentivize the entry of non-null agents. This will be the subject of Section~\ref{sec:nonnull-participation}.

\section{Continuous-Closure for General Losses}
\label{sec:general-losses}

We begin with a general principle for constructing post-hoc risk-safe families of mechanisms for arbitrary nonnegative losses. In particular, we introduce the continuous-closure region, whose selections satisfy \ref{eq:ic} and post-hoc risk control.
Optimizing over this can be computationally challenging, however, as it imposes a constraint for every subset of agents. This motivates the specializations and approximations developed in Section~\ref{sec:nap}.

To introduce the continuous-closure, we require the concept of an \emph{e-collection}~\citep{shafer2021testing,vovk2021values,grunwald2024beyond} which is a collection $\mathcal E=(E_S)_{S\subseteq[m]}$ of nonnegative random variables such that $E_S$ is an e-value under every joint law for which all agents in $S$ are null. Thus $\E[E_S]\leq1$ whenever $S$ consists entirely of null agents. We write $E_i=E_{\{i\}}$ for its singleton elements. Given an e-collection, define the set
\begin{equation}
\label{eq:implementable-allocation-set}
\begin{aligned}
    \mathcal X(\mathcal E)=\big\{A\in\Re_{\geq0}^m: A_i\leq c_iE_i, \;
      \|A\|_1 \leq B\big\},
\end{aligned}
\end{equation}
which is simply a set of allocations which satisfy \ref{eq:ic} and the budget constraint, given a particular e-collection.
Note that for any measurable selection $A\in\calX(\calE)$ and $i\in N$ we have $\E_N[A_i] \leq c_i \E_N[E_i] \leq c_i$.
To obtain post-hoc risk control in addition to null deterrence, given a loss $\ell$, define the \emph{continuous-closure} of $\calX$,
\begin{equation}
\label{eq:continuous-closure-region}
\begin{aligned}
    \mathcal A^\ell_\Delta(\mathcal E)
    =\big\{A\in\mathcal X(\mathcal E):\ell_S(A)\leq\Delta E_S
      \text{ for every }S\subseteq[m]\big\}.
\end{aligned}
\end{equation}
The superscript $\ell$ specifies the loss family to be used in the subset constraints.
The continuous-closure might appear unintuitive at first glance, but it is in fact close to the first thing one might try when attempting to control the expected loss. Indeed, note that for any $A\in\calA^\ell_\Delta(\calE)$, we have $\ell_N(A) \leq \Delta E_N$ by taking $S=N$. Hence $\E_N[\ell_N(A)] \leq \Delta \E_N[E_N] \leq \Delta$.

To control multiple losses simultaneously we simply gather the set of allocations which work for each loss.  This gives the following result.

\begin{theorem}
\label{thm:continuous-closure}
Fix an e-collection $\mathcal E$ and a family of losses $\calL$. Every measurable selection $A\in\calX(\calE)$ satisfies null deterrence. Moreover, for every null set $N\subseteq[m]$,
\begin{equation}
\label{eq:continuous-closure-expectation}
    \E_N\left[
      \sup_{\ell\in\mathcal L}\sup_{\Delta>0}\sup_{A\in\calA^\ell_\Delta(\calE)}
      \frac{\ell_N(A)}{\Delta}
    \right]\leq1,
\end{equation}
where $\E_N$ is understood as an outer expectation if the supremum is not measurable. In particular, if $A^{\ell,\Delta}$ is a measurable selection from $\calA^\ell_\Delta(\calE)$ for each $\ell\in\calL$ and $\Delta>0$, then
\begin{equation}
\label{eq:certified-menu-risk}
    \E_N\left[
      \sup_{\ell\in\mathcal L}\sup_{\Delta>0}
      \frac{\ell_N(A^{\ell,\Delta})}{\Delta}
    \right]\leq1.
\end{equation}
\end{theorem}

As in~\eqref{eq:expected-risk}, we emphasize that in both~\eqref{eq:continuous-closure-expectation} and~\eqref{eq:certified-menu-risk}, the set $N$ of null agents remains fixed throughout the suprema and the
expectation.

\begin{proof}[Proof of Theorem~\ref{thm:continuous-closure}]
For any measurable selection $A\in\calX(\calE)$ we have $A_i\leq c_iE_i$, which implies null deterrence after taking expectations. Moreover, for any $\ell\in\mathcal L$, $\Delta>0$, and $A\in\calA_\Delta^\ell(\calE)$, the constraint indexed by the true null set gives $\ell_N(A)/\Delta \leq E_N$. Since the right-hand side does not depend on $\ell$, $\Delta$, or $A$, taking the supremum over all three gives
\[
    \sup_{\ell\in\mathcal L}\sup_{\Delta>0}\sup_{A\in\calA^\ell_\Delta(\calE)}
      \frac{\ell_N(A)}{\Delta}\leq E_N,
\]
for every realization of $\calE$. Taking expectations and using $\E_N[E_N]\leq1$ gives~\eqref{eq:continuous-closure-expectation}. Finally, $\ell_N(A^{\ell,\Delta})/\Delta$ is bounded by the inner supremum in~\eqref{eq:continuous-closure-expectation} for every $\ell$ and $\Delta$, which proves~\eqref{eq:certified-menu-risk}.
\end{proof}

\begin{remark}
\label{rem:same-E}
    In the post-hoc case, using the same $\calE$ for every $\ell$ and $\Delta$ is essential in Theorem~\ref{thm:continuous-closure}.
    If each loss $\ell$ instead uses a separately valid collection $\mathcal E^\ell$, for example, the bound defining $\calA^\ell_\Delta(\calE^\ell)$ becomes $\ell_N(A^{\ell,\Delta})/\Delta\leq E_N^\ell$, and $\E_N[\sup_{\ell\in\mathcal L}E_N^\ell]$ need not be at most one. The allocation $A^{\ell,\Delta}$, on the other hand, can (and typically will) change with these parameters.
\end{remark}

The continuous-closure region is so named because it can be seen as a generalization of the closure principle in multiple testing~\citep{xu2025bringing}, in which we select a rejection set whose loss is bounded by an e-value for every possible null set. Of course, the choice in multiple testing is a binary one---we either reject the hypothesis or not. Here we need to decide how much to allocate to every agent, which significantly enlarges the decision space.

While Theorem~\ref{thm:continuous-closure} provides allocation regions whose selections yield post-hoc risk-safe families of mechanisms, these regions are not efficiently computable in general. The next step is thus to provide tractable approximations for specific loss functions. Going forward, we consider the null allocation rate.

\begin{remark}[On the choice of e-value]
\label{rem:e-value-choice}
Following \citet{bates2022principal}, the principal could allow each agent to select from a menu of e-values for $\calP$ before observing their data, with all guarantees unchanged.
Non-null agents could then choose an e-value suited to their data-generating process if they know it; powerful e-values are available for a wide range of settings~\citep{ramdas2025hypothesis}. We emphasize, however, that this choice must be made \emph{before} the data are generated, following Remark~\ref{rem:same-E}.
\end{remark}

\section{Optimal and Approximate Allocations for NAP}
\label{sec:nap}

In this section we focus specifically on $\ell=\nap$. Section~\ref{sec:nap-exact} shows that, for a suitable e-collection, the continuous-closure region is a polytope over which the principal can optimize in polynomial time, using a separation oracle based on max-flow/min-cut arguments. The resulting optimal allocation cannot generally be written in closed form, so Section~\ref{sec:nar-control} complements this procedure with simpler explicit mechanisms. Arguably, such mechanisms are more suitable for public protocols, since legibility engenders trust.
Section~\ref{sec:computational-comparison} compares the two approaches experimentally.

\subsection{Exact optimization in polynomial time}
\label{sec:nap-exact}

Our goal is to turn the region $\calA_\Delta^\nap$ into a convex region over which linear objectives can be optimized via linear programming. To do this, we will need to consider a particular e-collection. Let $E_i$ be an e-value for agent $i$, and define the \emph{cost-weighted e-collection} $\calE^{\mathrm w}$ by
\begin{equation}
\label{eq:cost-weighted-e-collection}
    E_S^\w=\frac{1}{C_S}\sum_{i\in S}c_iE_i
    \text{ for nonempty }S\subseteq[m],\qquad C_S = \sum_{i\in S}c_i.
\end{equation}
It's clear that $(E_S^\w)$ is an e-collection after taking expectations. We define $E_\emptyset^\w = 1$.

Fix a set of e-values $E_1,\dots,E_m$ and $\Delta>0$. For allocations $A_1,\dots,A_m$, write $\|A\|_1 = \sum_{i\in [m]}A_i$ for the total. The constraint defining $\calA_\Delta^\nap(\calE^\w)$ can be written as
\begin{equation}
\label{eq:pointwise-nap-certificate}
    C_S\sum_{i\in S}A_i\leq\Delta \|A\|_1\sum_{i\in S}c_iE_i,
\end{equation}
which is linear in the allocations. Thus, for realized e-values $E=(E_1,\dots,E_m)$, $\calA_\Delta^\nap(\calE^\w)$ is a convex polytope, which we denote by $\P_\Delta(E)$:
\begin{align}
    \label{eq:nap-polytope}
\P_\Delta(E) := \calA^\nap_\Delta(\mathcal E^\w) =
    \bigg\{A\in\calX(\calE^\w):  \sum_{i\in [m]} u_i(S) A_i \leq 0 \text{ for all } S\subseteq[m]
      \bigg\},
\end{align}
where $u_i(S) = \ind{i\in S} C_S - \Delta \sum_{j\in S} c_j E_j$ for all $i$. Note that $\P_\Delta(E)$ is convex because all its defining constraints are linear in $A$ for fixed $E$ and $\Delta$. Further, since allocating zero satisfies all the constraints, the polytope is nonempty.

Now, by Theorem~\ref{thm:continuous-closure}, the principal can choose any allocation in $\P_\Delta(E)$ and it will be post-hoc risk safe and satisfy~\eqref{eq:ic}. But, as we mentioned in Section~\ref{sec:contract}, the principal may also have an objective function $V(A;E)$ that it wishes to optimize over the polytope:
\begin{equation}
\label{eq:safe-allocation-optimization}
    \max_{A\in\P_\Delta(E)}V(A;E).
\end{equation}
If $V$ is linear (e.g., $V(A;E) = \sum_i A_i$) then this optimization problem is a linear program. A concave objective (e.g., $V(A;E) = \sum_i \log(1 + A_i)$) gives a convex optimization problem.

At first glance~\eqref{eq:safe-allocation-optimization} may appear intractable since $\P_\Delta(E)$ has one constraint for every subset of agents, and thus cannot be written down explicitly in polynomial time. Nevertheless, if $V$ is linear, polynomial-time optimization is possible:

\begin{proposition}
\label{prop:polytope-separation}
For fixed $E$ and $\Delta$, the polytope $\P_\Delta(E)$ admits a polynomial-time separation oracle. Consequently, linear optimization over $\P_\Delta(E)$ is polynomial-time solvable.
\end{proposition}

The proof is in Appendix~\ref{app:proofs}. The separation oracle finds a subset maximizing the violation of~\eqref{eq:pointwise-nap-certificate}. The intuition is as follows: The negative of this violation is a  submodular function, which can be represented as an $s$--$t$ cut objective~\citep{kolmogorov2002energy}. Thus a maximum-flow computation yields a minimum cut identifying a  violated constraint, or else certifies that all subset constraints hold (see~\citet{ahuja1993network} for background on max-flow/min-cut algorithms). The graph has $m+2$ vertices and $O(m^2)$ edges, so the maximum-flow algorithm gives a bound of $O(m^3)$ arithmetic operations per separation call. For rational input data, the resulting linear optimization guarantee is polynomial in $m$ and the input bit length.

More generally, maximizing a concave objective over $\P_\Delta(E)$ is a convex optimization problem, and approximate optimization is possible via the ellipsoid method. We omit a longer discussion here, but suffice it to say that under standard assumptions, the running time is polynomial in the input size and $\log(1/\epsilon)$ for an $\epsilon$-optimal solution.

Finally, let us provide necessary and sufficient conditions on when $\P_\Delta(E)$ contains any nonzero allocation.

\begin{proposition}
\label{prop:polytope-threshold}
The polytope $\P_\Delta(E)$ contains a nonzero allocation if and only if
\begin{equation}
\label{eq:polytope-global-threshold}
    \sum_{i=1}^m c_i E_i\geq \frac{C}{\Delta}.
\end{equation}
\end{proposition}
\begin{proof}
If $A\in\P_\Delta(E)$ is nonzero, constraint~\eqref{eq:pointwise-nap-certificate} with $S=[m]$ reads $C\|A\|_1\leq\Delta \|A\|_1\sum_{i}c_iE_i$, and dividing by $\|A\|_1>0$ gives~\eqref{eq:polytope-global-threshold}. Conversely, suppose~\eqref{eq:polytope-global-threshold} holds and let $G_i=c_iE_i$, so that $G\neq0$ and $C\leq\Delta \|G\|_1$. For every nonempty $S$, $C_S\sum_{i\in S}G_i\leq C\sum_{i\in S}G_i\leq\Delta \|G\|_1\sum_{i\in S}c_iE_i$, so $G$ satisfies every subset constraint. These constraints are invariant under positive scaling, so $tG$ satisfies them for $t=\min\{1,B/\|G\|_1\}$, which also satisfies the coordinate caps and the budget. Hence $tG$ is a nonzero element of $\P_\Delta(E)$.
\end{proof}

\subsection{Approximate-but-explicit policies for NAR control}
\label{sec:nar-control}

As we've just seen, the polytope $\P_\Delta(E)$ allows the principal to choose allocations by solving an optimization problem. The mechanisms in this subsection instead provide explicit policies which base allocations on simple greedy or threshold rules. Each nevertheless corresponds to optimizing a particular objective or imposing a particular constraint, which we make precise in Proposition~\ref{prop:implicit-objectives}. All three policies are built around a simple sufficient condition for membership in $\P_\Delta$,
formalized by defining the following subset:
\begin{equation}
\label{eq:compact-safe-allocation-polytope}
\Q_\Delta(E):=\bigg\{A\in\Re_{\geq0}^m: A_i\leq c_i E_i, \;
\; CA_i\leq \Delta c_i E_i \|A\|_1, \; \|A\|_1\leq B
\bigg\},
\end{equation}
where $C = \sum_{i\in[m]} c_i$ and again, $\|A\|_1 = \sum_{i\in[m]}A_i$. The polytope $\Q_\Delta$ can be seen as replacing the exponentially many subset constraints (i.e., one constraint for every subset $S\subseteq[m]$) with a single constraint per agent. Note that $\Q_\Delta$ need not be defined by an e-collection, but simply by an e-value for each agent. Following Remark~\ref{rem:e-value-choice}, each agent can choose the functional form of their e-value from
a menu of e-values for $\calP$, as long as the realized statistic $E_i = E_i(X_i)$ is computed after this choice is made.

Since $\Q_\Delta(E)\subseteq\P_\Delta(E)$, it's clear that selecting allocations from these regions yields a family of mechanisms satisfying the post-hoc risk guarantee of Theorem~\ref{thm:continuous-closure}. However, the nature of $\Q_\Delta(E)$ allows us to give the following more precise bound on the post-hoc risk. The proof is in Appendix~\ref{app:proofs}.

\begin{proposition}
\label{prop:compact-safe-allocation-polytope}
For every $E$ and $\Delta$, $\Q_\Delta(E)\subseteq\P_\Delta(E)$. Any mechanism that selects an allocation from $\Q_\Delta(E)$ for every $\Delta>0$ satisfies \ref{eq:ic} and $\risk_N(\calM,\nap)\leq C_N/C\leq 1$, where $C_N =\sum_{i\in N} c_i$ is the total cost of null agents, and $C = \sum_{i\in [m]}c_i$ is the total cost of all agents.
\end{proposition}

We now give several mechanisms based on $\Q_\Delta(E)$. The cumulative-evidence, all-or-nothing, and e-BH policies introduced below each select a particular point of $\Q_\Delta(E)$. Alternatively, of course, the principal can optimize a linear objective directly over $\Q_\Delta(E)$, which requires a single linear program with $O(m)$ constraints. Section~\ref{sec:computational-comparison} shows that this loses little relative to optimizing over $\P_\Delta(E)$. Still, there are reasons to present explicit mechanisms as well. For one, they are more interpretable. Second, as we'll see in Section~\ref{sec:nonnull-participation}, they allow for an analysis of the expected utility of non-null agents who opt into the contract.

\paragraph{The cumulative-evidence mechanism.}
Call $F_i = c_iE_i$ agent $i$'s \emph{potential allocation}. It is the maximum allocation allowed by $\calX(\calE)$.
We begin with a policy that funds the smallest  number of agents whose aggregate potential allocation reaches the global evidence threshold in~\eqref{eq:polytope-global-threshold}. (Note that if this threshold is not met, then $\P_\Delta(E)=\Q_\Delta(E)=\{0\}$.)

\begin{tcolorbox}[
    colback=pink,
    colframe=white,
    boxrule=0pt,
    sharp corners,
    left=6pt,
    right=6pt,
    top=6pt,
    bottom=6pt,
    breakable
]
\textbf{Cumulative-evidence mechanism $\Mcum_\Delta$}
\begin{enumerate}[nosep,topsep=0.1cm]
    \item Given $X_1,\dots,X_m$, compute $E_i=E_i(X_i)$ and $F_i=c_iE_i$, and sort the potential allocations as $F_{(1)}\geq\cdots\geq F_{(m)}$, breaking ties arbitrarily.
    \item If $\sum_{i=1}^mF_i<C/\Delta$, set $\tilde A_i=0$ for every agent. Otherwise, define
    \[
        k^{\rm cum}=\min\bigg\{k\in[m]:
        \sum_{j=1}^kF_{(j)}\geq\frac C\Delta\bigg\}.
    \]
    Set $\tilde A_{(j)}=F_{(j)}$ for $j<k^{\rm cum}$, and $\tilde A_{(k^{\rm cum})}
    =\frac C\Delta-\sum_{j<k^{\rm cum}}F_{(j)}$. Set all remaining allocations to zero.
    \item\label{m1} Write $W=\sum_i\tilde A_i$. If $W>B$, set $A_i=B\tilde A_i/W$; otherwise set $A_i=\tilde A_i$.
\end{enumerate}
\end{tcolorbox}

The final step rescales the raw allocations  when they exceed the budget. Since $A_i\leq\tilde A_i$, such rescaling preserves null deterrence, and multiplying everything by a common factor leaves the NAP unchanged. Appendix~\ref{sec:finite-budget-methods} discusses an alternative scaling method.


\paragraph{The all-or-nothing mechanism.}
For comparison, we next give a very simple mechanism which either allocates every agent their potential allocation or allocates to no one. We call this the all-or-nothing mechanism, $\Mall$.

\begin{tcolorbox}[
    colback=pink,
    colframe=white,
    boxrule=0pt,
    sharp corners,
    left=6pt,
    right=6pt,
    top=6pt,
    bottom=6pt,
    breakable
]
\textbf{All-or-nothing mechanism $\Mall_\Delta$}
\begin{enumerate}[nosep,topsep=0.1cm]
    \item Given $X_1,\dots,X_m$, compute $E_i=E_i(X_i)$ and $F_i=c_iE_i$.
    \item If $\sum_{i\in[m]} F_i \geq C/\Delta$, set $\tilde A_i = F_i$ for all $i\in[m]$. Otherwise set $\tilde A_i = 0$.
    \item Write $W=\sum_i\tilde A_i$. If $W>B$, set $A_i=B\tilde A_i/W$; otherwise set $A_i=\tilde A_i$.
\end{enumerate}
\end{tcolorbox}

Note that the activation condition is the same as $\Mcum$ and is exactly condition~\eqref{eq:polytope-global-threshold}. This may seem surprising at first glance. Isn't allocating to every agent demanding? One might expect the threshold to be much stronger than the minimum one. But $\Mall$ allocates in proportion to the potential allocations, so agents with weak evidence receive correspondingly little.

The simplicity of $\Mall$ comes with downsides. In particular, since all agents receive an allocation once the threshold is met, in situations where there are a few agents with extremely large potential allocations and a large number of agents with small to moderate potential allocations, all agents in the latter group receive an allocation. If mistaken allocations are especially costly, we might wish to be more conservative.

Motivated by this possible downside, we now present a more selective mechanism $\Mbh$ which, in these situations, would allocate only to agents with larger potential allocations.

\paragraph{The e-BH mechanism.}
This mechanism draws on the e-BH procedure of \citet{wang2022false}, which in turn adapts the Benjamini--Hochberg (BH) procedure~\citep{benjamini1995controlling} to e-values.
\begin{tcolorbox}[
    colback=pink,
    colframe=white,
    boxrule=0pt,
    sharp corners,
    left=6pt,
    right=6pt,
    top=6pt,
    bottom=6pt,
    breakable
]
\textbf{e-BH mechanism $\Mbh_\Delta$}
\begin{enumerate}[nosep,topsep=0.1cm]
    \item Given $X_1,\dots,X_m$, compute $E_i=E_i(X_i)$ and $F_i=c_iE_i$.
    \item Let $F_{(1)} \geq F_{(2)}\geq \dots\geq F_{(m)}$ be the order statistics. Define
    \begin{equation}
        k^*=\max\left\{k\in[m]: kF_{(k)}\geq \frac{C}{\Delta}\right\},\text{ with }\max\emptyset:=0.
    \end{equation}
    \item For all $i\leq k^*$, set $\tilde A_{(i)} = F_{(i)}$. For all other $i$, set $\tilde A_{(i)} = 0$.
    \item Write $W=\sum_i\tilde A_i$. If $W>B$, set $A_i=B\tilde A_i/W$; otherwise set $A_i=\tilde A_i$.
\end{enumerate}
\end{tcolorbox}

When all participation costs are equal, the selection condition becomes $kE_{(k)}\geq m/\Delta$, exactly the threshold used by the e-BH multiple-testing procedure at level $\Delta$. Thus, in this case, the funded agents coincide with the rejected hypotheses. For unequal costs, selection corresponds to weighted e-BH with weights $w_i=mc_i/C$~\citep{wang2022false}. The difference, of course, lies in the allocation: multiple testing returns binary rejection decisions, whereas our mechanism assigns each selected agent $c_iE_i$, followed by budget rescaling if necessary.

Neither $\Mall$ nor $\Mbh$ uniformly dominates the other. The AoN mechanism always activates when e-BH does, but e-BH withholds allocation from agents whose potential allocations are too small to satisfy the threshold. Whether this is preferable depends on the principal's objective.

Now let's turn to formal guarantees provided by these three mechanisms. We begin by observing that they do indeed belong to $\Q_\Delta(E)$, thereby inheriting the \ref{eq:ic} and post-hoc risk guarantees of Proposition~\ref{prop:compact-safe-allocation-polytope}. The proof is in Appendix~\ref{app:proofs}.
Figure~\ref{fig:mechanisms-polytope} illustrates Proposition~\ref{prop:mechanisms-in-polytope} for three agents.

\begin{proposition}
\label{prop:mechanisms-in-polytope}
For every realization $E$, level $\Delta>0$, and budget $B\in(0,\infty]$, the final allocations produced by $\Mcum_\Delta$, $\Mall_\Delta$, and $\Mbh_\Delta$ belong to $\Q_\Delta(E)$, and hence to $\P_\Delta(E)$.
\end{proposition}

\begin{figure}[t]
    \centering
\definecolor{polyblue}{RGB}{46,96,160}
\definecolor{polyred}{RGB}{176,52,42}
\begin{tikzpicture}[x=4.2cm, y=4.2cm, line join=round]
  \fill[polyblue!6] (0.766,-0.241) -- (0.000,0.000) -- (0.000,0.834) -- (0.766,0.594) -- cycle;
  \draw[polyblue!45, dashed, line width=0.4pt] (0.766,-0.241) -- (0.000,0.000) -- (0.000,0.834) -- (0.766,0.594) -- cycle;
  \fill[polyblue!6] (0.000,0.834) -- (0.000,0.000) -- (-0.257,-0.022) -- (-0.257,0.720) -- cycle;
  \draw[polyblue!45, dashed, line width=0.4pt] (0.000,0.834) -- (0.000,0.000) -- (-0.257,-0.022) -- (-0.257,0.720) -- cycle;
  \fill[polyblue!6] (0.509,-0.356) -- (-0.181,-0.139) -- (0.000,0.000) -- (0.766,-0.241) -- cycle;
  \draw[polyblue!45, dashed, line width=0.4pt] (0.509,-0.356) -- (-0.181,-0.139) -- (0.000,0.000) -- (0.766,-0.241) -- cycle;
  \fill[polyblue!6] (-0.257,-0.022) -- (0.000,0.000) -- (-0.181,-0.139) -- cycle;
  \draw[polyblue!45, dashed, line width=0.4pt] (-0.257,-0.022) -- (0.000,0.000) -- (-0.181,-0.139) -- cycle;
  \draw[-stealth, line width=0.5pt] (0.000,0.000) -- (-0.399,-0.178) node[pos=1.21] {$A_1$};
  \draw[-stealth, line width=0.5pt] (0.000,0.000) -- (0.935,-0.294) node[pos=1.09] {$A_2$};
  \draw[-stealth, line width=0.5pt] (0.000,0.000) -- (0.000,1.038) node[pos=1.09] {$A_3$};
  \fill[polyred!37!white] (0.509,0.201) -- (0.279,0.551) -- (-0.104,0.672) -- (0.509,-0.263) -- cycle;
  \draw[polyred, line width=0.6pt] (0.509,0.201) -- (0.279,0.551) -- (-0.104,0.672) -- (0.509,-0.263) -- cycle;
  \fill[polyred!29!white] (0.509,0.201) -- (0.509,-0.263) -- (0.766,0.223) -- (0.766,0.594) -- (0.702,0.565) -- cycle;
  \draw[polyred, line width=0.6pt] (0.509,0.201) -- (0.509,-0.263) -- (0.766,0.223) -- (0.766,0.594) -- (0.702,0.565) -- cycle;
  \fill[polyred!59!white] (0.702,0.565) -- (0.766,0.594) -- (0.460,0.690) -- (-0.104,0.672) -- (0.279,0.551) -- cycle;
  \draw[polyred, line width=0.6pt] (0.702,0.565) -- (0.766,0.594) -- (0.460,0.690) -- (-0.104,0.672) -- (0.279,0.551) -- cycle;
  \fill[polyred!56!white] (0.509,0.201) -- (0.702,0.565) -- (0.279,0.551) -- cycle;
  \draw[polyred, line width=0.6pt] (0.509,0.201) -- (0.702,0.565) -- (0.279,0.551) -- cycle;
  \fill[polyred!22!white] (-0.104,0.672) -- (0.000,0.000) -- (0.509,-0.263) -- cycle;
  \draw[polyred, line width=0.6pt] (-0.104,0.672) -- (0.000,0.000) -- (0.509,-0.263) -- cycle;
  \fill[polyblue!27!white, fill opacity=0.45] (0.279,0.551) -- (-0.257,0.720) -- (-0.257,-0.022) -- (-0.181,-0.139) -- (0.509,-0.356) -- (0.509,0.201) -- cycle;
  \fill[polyblue!19!white, fill opacity=0.45] (0.509,0.201) -- (0.509,-0.356) -- (0.766,-0.241) -- (0.766,0.594) -- (0.702,0.565) -- cycle;
  \fill[polyblue!49!white, fill opacity=0.45] (0.702,0.565) -- (0.766,0.594) -- (0.000,0.834) -- (-0.257,0.720) -- (0.279,0.551) -- cycle;
  \fill[polyblue!46!white, fill opacity=0.45] (0.509,0.201) -- (0.702,0.565) -- (0.279,0.551) -- cycle;
  \draw[polyblue, line width=0.6pt] (0.279,0.551) -- (-0.257,0.720) -- (-0.257,-0.022) -- (-0.181,-0.139) -- (0.509,-0.356) -- (0.509,0.201) -- cycle;
  \draw[polyblue, line width=0.6pt] (0.509,0.201) -- (0.509,-0.356) -- (0.766,-0.241) -- (0.766,0.594) -- (0.702,0.565) -- cycle;
  \draw[polyblue, line width=0.6pt] (0.702,0.565) -- (0.766,0.594) -- (0.000,0.834) -- (-0.257,0.720) -- (0.279,0.551) -- cycle;
  \draw[polyblue, line width=0.6pt] (0.509,0.201) -- (0.702,0.565) -- (0.279,0.551) -- cycle;
  \draw[polyred, line width=0.8pt] (0.509,0.201) -- (0.279,0.551) -- (-0.104,0.672) -- (0.509,-0.263) -- cycle;
  \draw[polyred, line width=0.8pt] (0.509,0.201) -- (0.509,-0.263) -- (0.766,0.223) -- (0.766,0.594) -- (0.702,0.565) -- cycle;
  \draw[polyred, line width=0.8pt] (0.702,0.565) -- (0.766,0.594) -- (0.460,0.690) -- (-0.104,0.672) -- (0.279,0.551) -- cycle;
  \draw[polyred, line width=0.8pt] (0.509,0.201) -- (0.702,0.565) -- (0.279,0.551) -- cycle;
  \draw[polyred, line width=0.8pt] (-0.104,0.672) -- (0.000,0.000) -- (0.509,-0.263) -- cycle;
  \draw[polyblue, line width=0.4pt] (-0.181,0.529) -- (-0.357,0.609) node[left] {$\P_\Delta(E)$};
  \draw[polyred, line width=0.4pt] (0.135,0.136) -- (-0.357,0.086) node[left] {$\Q_\Delta(E)$};
  \draw[-stealth, line width=0.6pt, shorten <=2.5pt] (0.766,0.223) -- (1.016,0.223) node[right] {$\Mcum_\Delta$};
  \draw[-stealth, line width=0.6pt, shorten <=2.5pt] (0.443,0.416) -- (1.016,0.416) node[right] {$\Mall_\Delta$};
  \draw[-stealth, line width=0.6pt, shorten <=2.5pt] (0.766,0.594) -- (1.016,0.594) node[right] {$\Mbh_\Delta$};
  \fill[black, draw=white, line width=0.4pt] (0.766,0.223) circle[radius=1.5pt];
  \fill[black, draw=white, line width=0.4pt] (0.443,0.416) circle[radius=1.5pt];
  \fill[black, draw=white, line width=0.4pt] (0.766,0.594) circle[radius=1.5pt];
\end{tikzpicture}
    \caption{The polytopes $\P_\Delta(E)$ (blue) and $\Q_\Delta(E)$ (red) for $m=3$ agents with $c_i=1$, $E=(8,20,18)$, $\Delta=0.1$, and $B=40$, along with the allocations selected by $\Mcum_\Delta$, $\Mall_\Delta$, and $\Mbh_\Delta$. Both regions satisfy the caps $A_i\leq c_iE_i$ and the budget.
    $\Mall_\Delta$ funds all three agents, $\Mbh_\Delta$ drops agent 1, whose evidence is weakest, and $\Mcum_\Delta$ funds agent 2 fully and gives agent 3 the remainder needed to reach $C/\Delta$.}
    \label{fig:mechanisms-polytope}
\end{figure}
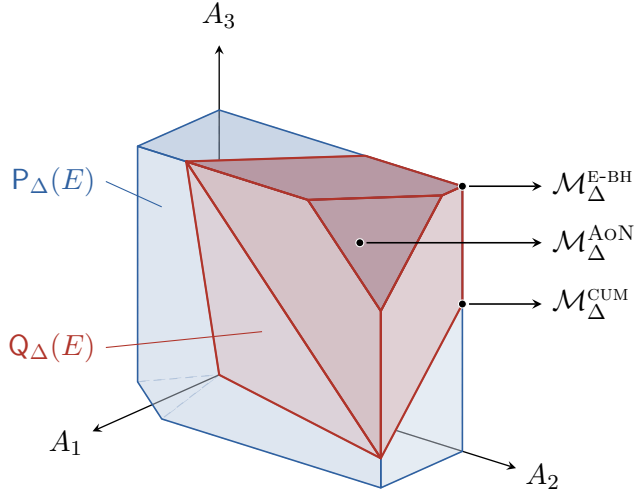


We stated above that the mechanisms $\Mcum$, $\Mall$, and $\Mbh$ are all optimal for a particular objective. To formalize this, we need to introduce two pieces of notation. First, let us call a set $T\subseteq[m]$ \emph{self-consistent} if $F_i\geq C/(\Delta|T|)$ for every $i\in T$. Intuitively, each element $i$ of a self-consistent set $T$ must contribute at least an equal share of the threshold~\eqref{eq:polytope-global-threshold}, so larger sets permit smaller individual potential allocations. Second, define
\begin{equation}
    V_{\min}(A)=\min_{i:F_i>0}\frac{A_i}{F_i},
\end{equation}
which measures the smallest allocation-to-potential-allocation ratio among agents with positive potential allocations. We set $V_{\min}(A)=0$ when $F_i=0$ for every $i$, in which case the only feasible allocation is zero.

\begin{proposition}
\label{prop:implicit-objectives}
Fix $E$ and $\Delta$.
\begin{enumerate}[nosep,topsep=0.1cm]
    \item The allocation of $\Mall_\Delta$ is the unique maximizer of $V_{\min}$ over $\P_\Delta(E)$ and $\Q_\Delta(E)$. Furthermore, when $B=\infty$, it maximizes every coordinate-wise nondecreasing objective over both regions.
    \item Whenever $\Mcum_\Delta$ activates, its final allocation has minimum support size among all nonzero elements of $\Q_\Delta(E)$. Otherwise, $\Q_\Delta(E)=\{0\}$. It activates if and only if $\Mall$ activates. Moreover, whenever $\Mbh$ activates, $k^{\rm cum}\leq k^*$, but the converse activation implication need not hold.
    \item The raw allocation of $\Mbh_\Delta$ is $F\mathbf{1}_{T^*}$, where $T^*$ is the largest self-consistent set. In particular, every self-consistent $T$ satisfies $T\subseteq T^*$. Consequently, the allocation maximizes every coordinate-wise nondecreasing objective over $\{F\mathbf{1}_T:T\text{ self-consistent}\}$.
\end{enumerate}
\end{proposition}

The proof is in Appendix~\ref{app:proofs}. It's perhaps worth elaborating slightly on the above result. It shows that AoN is itself an instance of~\eqref{eq:safe-allocation-optimization}, with the concave objective $V_{\min}$, whose optimum happens to be available in closed form. The cumulative rule instead minimizes the number of funded agents. This objective is not concave, and the result holds over $\Q_\Delta(E)$ but not over $\P_\Delta(E)$. Finally, e-BH optimizes under a self-consistency constraint rather than a different objective. Moreover, when $B=\infty$, its allocation is coordinate-wise dominated by that of AoN, so its justification lies in excluding agents whose evidence is weak on its own.

\subsection{Experimental comparison}
\label{sec:computational-comparison}

We compare five mechanisms: optimization over $\P_\Delta(E)$, optimization over $\Q_\Delta(E)$, and the three explicit policies $\Mcum$, $\Mall$, and $\Mbh$. We measure their performance on two objective functions: the linear objective $V_{\rm lin}(A)=\sum_i v_i A_i$ and the concave objective $V_{\log}(A)=\sum_i v_i\log(1+A_i)$, where $v_i$ are observed positive weights. We sample these weights independently from a log-normal distribution with mean one and log-scale standard deviation $\sigma_v$. (The weights are, of course, held fixed during the optimization.) We vary $\sigma_v$ from $0$ (equal weights) to $1.2$, reusing the same evidence and participation costs at each value. We generate costs independently from log-normal distributions and use likelihood-ratio e-values. Details and further experiments are in Appendix~\ref{app:experiments}.

\begin{figure}[t]
    \centering
    \includegraphics[width=\textwidth]{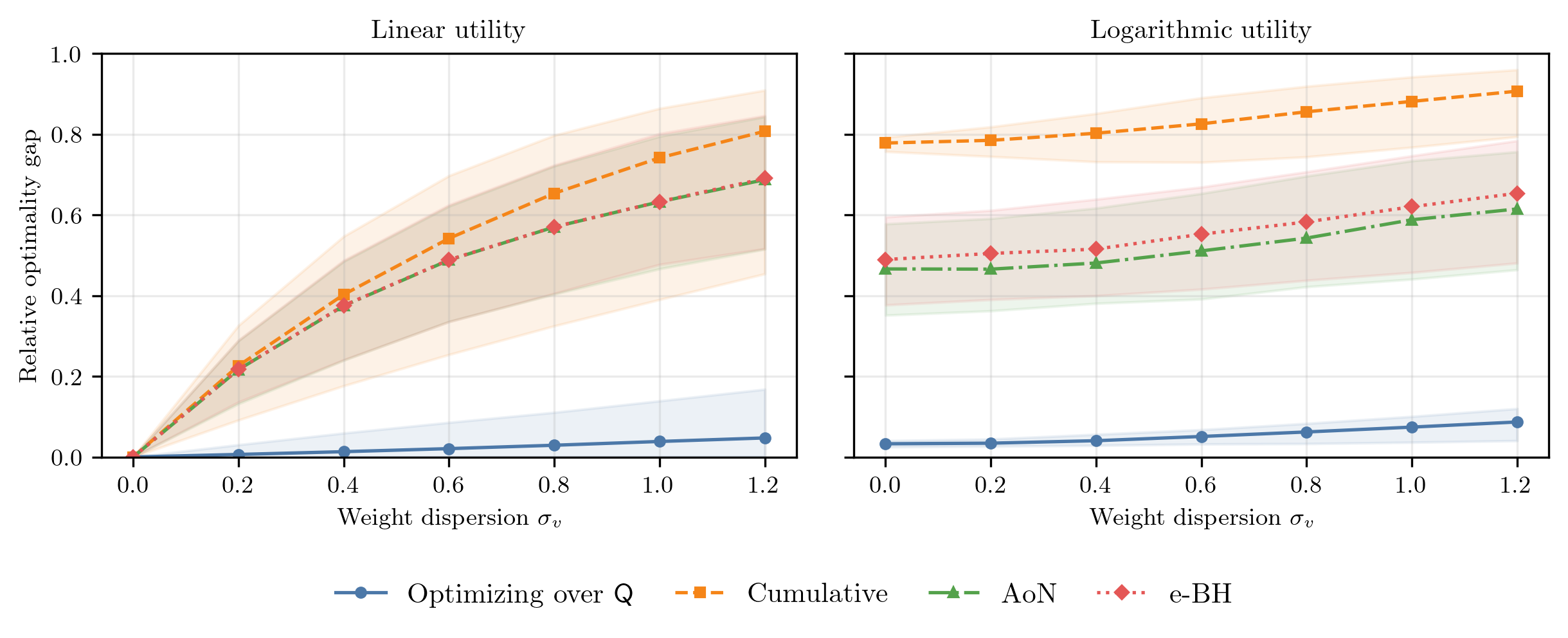}
    \caption{The relative optimality gap of various mechanisms compared with optimizing over the full continuous-closure region $\P_\Delta$, as the dispersion $\sigma_v$ of the weights varies, for linear utility (left) and logarithmic utility (right). Lines show medians and shaded bands show interquartile ranges over 150 independently generated active instances with $m=30$, $B=30$, and $\Delta=0.1$. The same instances are used at every value of $\sigma_v$, and $\sigma_v=0$ corresponds to equal weights.}
    \label{fig:optimality-gap-runtime}
\end{figure}

Figure~\ref{fig:optimality-gap-runtime} shows that optimization over $\Q_\Delta$ loses little relative to optimizing over $\P_\Delta$ at every level of dispersion. Its median relative gap grows with $\sigma_v$, but is at most $4.7\%$ for the linear objective and $8.7\%$ for the concave objective. At $\sigma_v=0.8$, for instance, its median and 90th-percentile gaps are $2.9\%$ and $19.6\%$ for the linear objective, and $6.2\%$ and $11.6\%$ for the concave objective. The explicit policies do not use the weights (they are maximizing their own implicit objectives, as per Proposition~\ref{prop:implicit-objectives}) and their gaps grow more quickly as the weights become more dispersed. At $\sigma_v=0.8$, their median gaps range from $57.0\%$ to $65.4\%$ for the linear objective and from $54.3\%$ to $85.6\%$ for the concave objective.

When $\sigma_v=0$, the linear objective $\sum_i A_i$ depends only on the total allocation. Whenever a nonzero allocation is feasible, $\Mall$ attains $\min\{B,\sum_i F_i\}$, the largest total allowed by the budget and coordinate bounds, so it is optimal over both $\P_\Delta(E)$ and $\Q_\Delta(E)$. Every other feasible allocation spending the full budget is also optimal, even if it distributes resources very differently. In our instances, every mechanism is optimal in this case, except on a single instance on which $\Mbh$ does not activate. Heterogeneous weights are what make these differences matter for the linear objective. For the concave objective, by contrast, substantial gains from optimization remain even with equal weights.  Optimization over $\Q_\Delta$ has a median gap of $3.3\%$ and a 90th-percentile gap of $4.8\%$, while the median gaps of the explicit policies range from $46.6\%$ to $77.8\%$.

Intriguingly, the optimizers allocate more to null agents than the explicit policies do, but the amounts remain small. Across all values of $\sigma_v$, the average realized null allocation proportion over these instances is at most $3.9\%$ when optimizing over $\P_\Delta$ and $2.9\%$ when optimizing over $\Q_\Delta$, compared with at most $0.22\%$ for the explicit policies. All of these are well below $\Delta=0.1$.

Figure~\ref{fig:optimality-gap-scaling} examines how allocation quality changes with the number of agents, fixing $\sigma_v=0.8$. As $m$ increases from 10 to 100, optimizing over $\Q_\Delta$ has a median gap remaining below $8\%$ for both objectives, rising from $3.7\%$ to $6.8\%$ for linear utility and from $5.1\%$ to $7.5\%$ for concave utility. The explicit policies' gaps generally increase more substantially, reaching $67.7\%$--$73.2\%$ for linear utility and $64.5\%$--$91.7\%$ for concave utility at $m=100$. These comparisons keep $B=30$ fixed, so increasing the number of agents also increases competition for the same resources.

\begin{figure}[t]
    \centering
    \includegraphics[width=\textwidth]{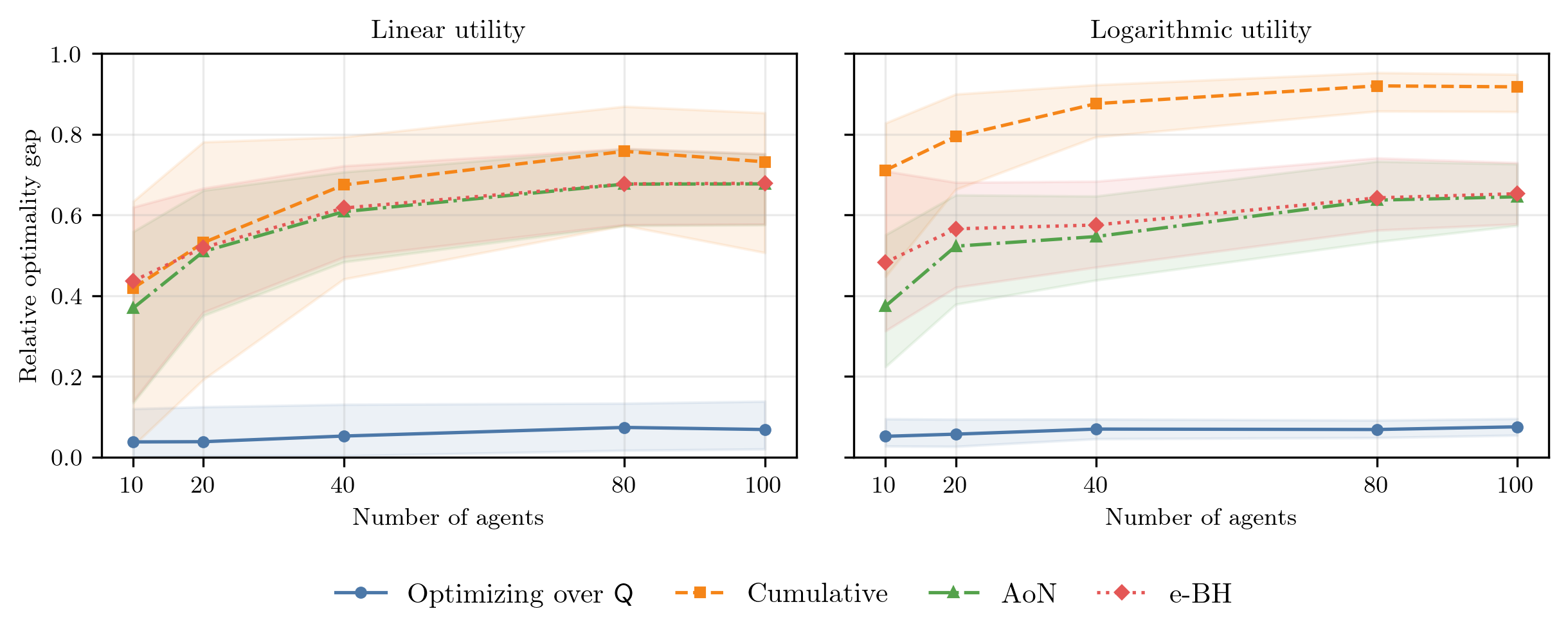}
    \caption{Relative optimality gap as the number of agents varies, for linear utility (left) and logarithmic utility (right). Lines show medians and shaded bands show interquartile ranges over 150 independently generated active instances at each size, with $B=30$, $\Delta=0.1$, and $\sigma_v=0.8$ fixed. As in Figure~\ref{fig:optimality-gap-runtime}, the gap is measured relative to optimization over $\P_\Delta$.}
    \label{fig:optimality-gap-scaling}
\end{figure}

Overall, we find that optimizing over $\Q_\Delta$ loses little relative to optimizing over $\P_\Delta$. The explicit mechanisms lose more, including under equal weights for the logarithmic objective. This is to be expected given Proposition~\ref{prop:implicit-objectives}, since each is maximizing a specific objective that is distinct from that of the principal.


\FloatBarrier
\section{Non-Null Participation}
\label{sec:nonnull-participation}

Null deterrence ensures that null agents cannot profit in expectation from entering the contract. At the same time, we also want to encourage the participation of non-null agents. One might hope that any non-null agent with a large e-value will be incentivized to participate, but this need not be the case, as the following example shows.

\begin{example}
\label{ex:participation-dilution}
Consider two agents, each collecting iid observations $X_{i,t}\sim\mathcal N(\mu_i,1)$, with null hypothesis $\mu_i=0$. The likelihood ratios
\begin{equation}
E_{i,n}=\exp\left(\sum_{t=1}^n X_{i,t}-n/2\right)
\end{equation}
are valid e-values under the relevant distributions. If $\mu_1=1$ and $\mu_2=2$, then $n^{-1}\log E_{1,n}\xrightarrow{a.s.}1/2$ and $n^{-1}\log E_{2,n}\xrightarrow{a.s.}3/2$. Both agents' evidence diverges, and both are selected by $\Mall$ and $\Mbh$ with probability tending to one. Nevertheless, $c_1E_{1,n}/(c_1E_{1,n}+c_2E_{2,n})\to0$. In other words, with a finite budget $B$, $A_{1,n}\to0$ almost surely.
\end{example}

In other words, strong evidence need not make participation in the contract profitable, as the total allocation can be dominated by an agent with even stronger evidence. There are several reasons we might like to avoid this scenario. For one, the principal may wish to avoid putting all its eggs in one basket if it can avoid it. In Example~\ref{ex:pharma}, for instance, the pharmaceutical company may wish to fund multiple startups in the interest of diversifying its risk. Second, it may be rational for each non-null agent to opt out of the contract, out of fear that another agent has a larger e-value.

We can address this problem by truncating each agent's e-value. The details are given in Section~\ref{sec:cap}. They apply only to two of the three explicit mechanisms for reasons which will be discussed below. Section~\ref{sec:lower_floor_opt} discusses how the principal can optimize over the polytopes $\P_\Delta$ or $\Q_\Delta$ while maintaining a floor on each agent's allocation.

\subsection{Participation under truncated evidence}
\label{sec:cap}

For a truncation level $L>0$, consider replacing each $E_i$ by $\bar E_i=\min\{E_i,L\}$ before running $\Mall$ or $\Mbh$. Since $\bar E_i\leq E_i$, the resulting  allocations still belong to $\Q_\Delta(\bar E)\subseteq\Q_\Delta(E)$, still satisfy \ref{eq:ic}, and still bound post-hoc risk at $C_N/C$. While we are permitted to choose a measurable truncation level $L=L(E,\Delta)$ based on the data, the analysis is cleaner if we consider deterministic truncation levels.

Our first result gives a sufficient condition for an agent to profit, in expectation, from the contract. The proof is in Appendix~\ref{app:proofs}.

\begin{proposition}
\label{prop:finite-sample-participation}
Fix a deterministic truncation level $L>0$ and a non-null agent $i$ with $c_iL\geq C/\Delta$. Write $\pi_i=\mathbb{P}_H(E_i\geq L)$, where $H$ is any joint law over the agents' data. Running $\Mall$ or $\Mbh$ with truncated e-values $\bar E_i$  gives allocations satisfying
\begin{equation}
\label{eq:finite-sample-payment-floor}
    \E_H[A_i]\geq\pi_i c_i\min\{L,B/C\}.
\end{equation}
Consequently, a sufficient condition for strictly profitable entry is
\begin{equation}
\label{eq:finite-sample-participation}
    \pi_i>\max\{1/L,C/B\}.
\end{equation}
\end{proposition}

Note that while an agent does not know the joint law $H$, $\pi_i$ depends only on the distribution of their own e-value. Thus, if $L$, $C$, and $B$ are public, an agent may  calculate whether~\eqref{eq:finite-sample-participation} holds. As usual, this result holds under arbitrary dependence between the agents. The condition $c_iL\geq C/\Delta$ in Proposition~\ref{prop:finite-sample-participation} is necessary. Indeed, suppose we fund an agent $i$ as soon as they reach the truncation level $L$. In the event that all other agents are null and have e-values of 0,~\eqref{eq:polytope-global-threshold} gives $c_i L \geq c_i\bar E_i\geq C/\Delta $.

Proposition~\ref{prop:finite-sample-participation} gives a lower bound on an agent's allocation. In the asymptotic regime, meanwhile, we can give a precise characterization of the expected allocation.

\begin{theorem}
\label{thm:capped-participation}
Fix $m$, $N$, $\Delta$, and a finite budget $B<\infty$, and suppose that $S=[m]\setminus N$ (the set of non-null agents) is nonempty. Let $H$ be any joint law over the agents' data under which $N$ is the set of null agents. Let $n$ be the number of observations in each agent's dataset. Suppose $E_{i,n}$ is a valid e-value for every $i$ and $n$, and let $L_n>0$ be deterministic with $L_n\to\infty$. Run $\Mall$ or $\Mbh$ on $\bar E_{i,n}=\min\{E_{i,n},L_n\}$. For any agent $i\in S$ satisfying
\begin{equation}
\label{eq:cap-reaching-power}
    \mathbb{P}_H(E_{i,n}\geq L_n)\longrightarrow1,
\end{equation}
the expected net payoff $U_{i,n}=\E_H[A_{i,n}]-c_i$ satisfies
\begin{equation}
\label{eq:capped-participation-lower-bound}
    \liminf_{n\to\infty}U_{i,n}
    \geq c_i\left(\frac{B}{C_S}-1\right).
\end{equation}
In particular, $B>C_S$ guarantees eventually strictly profitable entry for this agent. Moreover, if~\eqref{eq:cap-reaching-power} holds for every $i\in S$, then
\begin{equation}
\label{eq:capped-participation-limit}
    \lim_{n\to\infty}\E_H[A_{i,n}]
    =\begin{cases}
        Bc_i/C_S,&i\in S,\\
        0,&i\in N,
    \end{cases}
\end{equation}
and $\E_H[\nap_N(A_n)]\to0$.
\end{theorem}

The proof is in Appendix~\ref{app:proofs}. It's useful to observe that condition~\eqref{eq:cap-reaching-power} holds for agent $i$ if $n^{-1}\log E_{i,n}$ converges in $H$-probability to a positive constant and $L_n$ satisfies $\log L_n=o(n)$. Returning to Example~\ref{ex:participation-dilution}, if we truncate the agents' e-values appropriately, then by~\eqref{eq:capped-participation-limit}, agent 1 receives expected allocation $Bc_1/(c_1+c_2)$ in the limit despite agent 2's e-value growing at a faster exponential rate.

\begin{remark}
\label{rem:allocation-and-cost}
As is made clear by~\eqref{eq:capped-participation-limit} but also apparent in $\Mcum$, $\Mall$, and $\Mbh$, allocations scale with the entry cost $c_i$. In other words, two agents with identical e-values can receive substantially different payoffs if their participation costs differ. We might consider this only fair: a larger upfront investment results in a potentially larger payoff.

On the other hand, the principal might prefer that allocations reflect evidence alone. In this case, one could choose weights $0<\kappa_i\leq c_i$, replacing $c_i$ by $\kappa_i$ in the coordinate caps and the cost-weighted e-collection, and $C_S$ by $K_S=\sum_{i\in S}\kappa_i$. Equal weights make allocations depend on evidence alone. Since $A_i\leq\kappa_iE_i\leq c_iE_i$, null deterrence is preserved, and the post-hoc risk bound of Proposition~\ref{prop:compact-safe-allocation-polytope} becomes $K_N/K_{[m]}$. Payments then scale with $\kappa_i$ while entry still costs $c_i$, however, which changes the participation and payment guarantees.
\end{remark}

Proposition~\ref{prop:finite-sample-participation} and Theorem~\ref{thm:capped-participation} hold only for $\Mall$ and $\Mbh$. The mechanism $\Mcum$ need not obey their conclusions. For instance, if two agents both reach the truncation level (i.e., $E_i \geq L$) and $c_1>c_2$, then $\Mcum$ funds only agent 1 once $c_1L\geq C/\Delta$. In this case, agent 2 receives zero despite also reaching the truncation level.

Appendix~\ref{app:truncation-necessity} shows that there is a formal sense in which truncating the e-value is necessary to ensure that a non-null agent's payout is not diluted. In particular, if we do not want the limiting allocation to depend on the agent's rate of growth, then we must truncate the e-value. However, if we are prepared to allow those growth rates into the final allocation, then we can consider sublinear transformations of the e-values. The following remark gives an example.

\begin{remark}
\label{rem:log-thresholding}
As alluded to above, instead of truncating the e-value consider replacing $E_i$
by $\widetilde E_i=\log(1+E_i)$. Since $0\leq\widetilde E_i\leq E_i$, this also preserves null deterrence and the post-hoc risk bound. Fix $m$, $N$, $\Delta$, $B$, and $H$ as in Theorem~\ref{thm:capped-participation}, and suppose $n^{-1}\log E_{i,n}\to r_i\in(0,\infty)$ in $H$-probability for every $i\in S$. Running $\Mall$ or $\Mbh$ on the transformed evidence gives
\begin{equation}
    \lim_{n\to\infty}\E_H[A_{i,n}]
    =\frac{Bc_i r_i}{\sum_{j\in S}c_jr_j}\quad \text{ for all }i\in S.
\end{equation}
Indeed, $\widetilde E_{i,n}/n\to r_i$ for non-null agents, while $\widetilde E_{i,n}/n\to0$ for null agents. Thus every non-null agent is selected with probability tending to one. Consequently, logarithms prevent the vanishing share in Example~\ref{ex:participation-dilution}, but they cannot escape the dependence on the limiting growth rates.
Returning to Example~\ref{ex:participation-dilution}, agent 1's limiting payment in this scenario is $Bc_1/(c_1+3c_2)$, reflecting that $n^{-1} \log E_{1,n}\to 1/2$ and $n^{-1}\log E_{2,n} \to 3/2$.
\end{remark}



We next ask how the principal can optimize its own objective (over the safe allocation region $\P_\Delta(E)$) while preserving incentives for sufficiently powerful non-null agents to participate.

\subsection{Optimizing with payment floors}
\label{sec:lower_floor_opt}

For e-values $E=(E_1,\dots,E_m)$ and a level $\Delta>0$, let $b(E)=(b_1(E),\dots,b_m(E))$ be the allocations produced by $\Mall$ with truncated e-values. We consider the allocation of $\Mall$ the ``baseline payment.''
To protect a fraction $\lambda\in[0,1]$ of this baseline while optimizing the principal's objective, consider
\begin{equation}
\label{eq:payment-floor-optimization}
    \max_{A\in\P_\Delta(E)}V(A;E)\text{ subject to }A_i\geq\lambda b_i(E)\text{ for every }i\in[m].
\end{equation}

The following result bounds the loss in the principal's objective from imposing these payment floors.

\begin{proposition}
\label{prop:payment-floors}
If $V(\cdot;E)$ is continuous and concave, let $A^*$ maximize $V$ over $\P_\Delta(E)$ and let $A^\lambda$ solve~\eqref{eq:payment-floor-optimization}. For every $E$ and $\lambda\in[0,1]$,
\begin{equation}
\label{eq:payment-floor-objective-bound}
    V(A^*;E)-V(A^\lambda;E)
    \leq\lambda\bigl[V(A^*;E)-V(b(E);E)\bigr].
\end{equation}
\end{proposition}

The proof is in Appendix~\ref{app:proofs}, and shows that $\lambda$ may be data-dependent. But to obtain a simple lower bound on expected net payoffs, let us suppose that $\lambda$ is fixed. Under the setup of Theorem~\ref{thm:capped-participation}, with fixed costs and a fixed finite budget $B$, using its truncated AoN allocations as the baselines gives, for any agent $i\in S$ satisfying~\eqref{eq:cap-reaching-power},
\begin{equation}
\label{eq:payment-floor-participation}
    \liminf_{n\to\infty}
        \big(\E_H[A^\lambda_{i,n}]-c_i\big)
    \geq c_i\left(\frac{\lambda B}{C_S}-1\right).
\end{equation}
Together, these bounds quantify the price of protecting participation. If $B>C_S$, choosing $\lambda>C_S/B$ guarantees eventual strict profitability for this agent, while sacrificing at most a fraction $\lambda$ of the gap between the baseline and the optimum.

\section{Extensions to Online Arrival and Two-Stage Contracts}
\label{sec:extensions}

We next consider two steps toward a fully sequential allocation theory. First, agents approach the principal over time and receive irrevocable allocations. Second, funded agents use their allocations to conduct another round of experimentation and return with fresh evidence. Selective entry, in which agents observe their data before deciding whether to approach the principal, is treated separately in Appendix~\ref{app:selective-entry}.

\subsection{Online multi-agent statistical contracts}
We begin with the online setting. We allow a given agent to enter the contract based on what has happened thus far, including which other agents have entered, the data they've generated, and the allocations that have been made. The principal must decide on-the-fly how much to allocate to each agent that enters---it cannot decide to allocate later, nor revoke any allocation that has been made previously. We also allow the principal's budget to increase at every timestep.

\begin{tcolorbox}[
    colback=blue!8!white,
    colframe=white,
    boxrule=0pt,
    sharp corners,
    left=6pt,
    right=6pt,
    top=6pt,
    bottom=6pt,
    breakable
]
\textbf{Online multi-agent statistical contract}
\begin{enumerate}
    \item The principal announces mechanism $\calM$, a sequence of publicly known participation costs $(c_t)_{t\geq1}$ with $c_t>0$, and initial budget $B_1\in(0,\infty]$.
    \item For $t=1,2,\dots$
    \begin{enumerate}
        \item An agent with private distribution $P$ decides whether to pay $c_t$ to enter the contract. The agent has access to all prior information, including which agents have entered and what they've been allocated.
        \item Upon entry, data $X_t$ is generated from $P$. The principal feeds $X_t$ into $\calM$ which determines an allocation $A_t$ for the agent.
        \item At all times $t$, the allocations must satisfy $\sum_{i\leq t} A_i\leq B_t$.
        \item The total budget gets updated to $B_{t+1}\geq B_t$.
    \end{enumerate}
\end{enumerate}
\end{tcolorbox}

The guarantees introduced in Section~\ref{sec:contract} need to be modified somewhat in the online setting. In particular, since agents enter strategically, the set $N\subseteq\mathbb{N}$ of times at which a null agent enters may be random. Null deterrence is therefore conditional on agent $t$ entering the contract, and becomes
\begin{equation}
    \E_N[(A_t - c_t)\ind{t\in N}] \leq 0,
\end{equation}
which collapses to~\ref{eq:ic} in the batch setting. This in turn requires that each agent's e-value is valid conditional on the past. That is, letting $\calF_t$ denote the information available after the first $t$ participating agents have been processed by the contract, we will require the e-values to satisfy
\begin{equation}
\label{eq:online-null-validity}
    \E\big[E_t\ind{t\in N}\mid\calF_{t-1}\big]\leq\mathbb{P}(t\in N\mid\calF_{t-1})\text{ for every }t\geq1.
\end{equation}
This holds, for instance, if the type of the $t$-th agent is fixed before $X_t$ is drawn and $\E_\calP[E_t\mid\calF_{t-1}]\leq1$. In this section, $\E_N$ denotes the supremum over all joint laws satisfying~\eqref{eq:online-null-validity}.

Next we propose a mechanism that controls the NAR simultaneously at every time, including at stopping times.

\begin{tcolorbox}[
    colback=pink,
    colframe=white,
    boxrule=0pt,
    sharp corners,
    left=6pt,
    right=6pt,
    top=6pt,
    bottom=6pt,
    breakable
]
\textbf{Online mechanism $\Monline_\Delta$}
\begin{enumerate}[nosep,topsep=0.1cm]
    \item Set $W_0 = 0$ and let $(\gamma_t)_{t=1}^\infty$ be a predictable nonnegative process: $\gamma_t$ is $\calF_{t-1}$-measurable and $\sum_{t=1}^\infty\gamma_t\leq1$ almost surely.
    \item For $t=1,2,\dots$ \# Here $t$ indicates the $t$-th agent to have opted into the contract
    \begin{enumerate}
        \item Define $\beta_t = 1 - \Delta\gamma_tW_{t-1}/c_t$. Set $\Delta_t=\Delta\gamma_t/\beta_t$ if $\beta_t>0$, and set $\Delta_t=\infty$ otherwise, with the conventions $1/0=\infty$ and $1/\infty=0$.
        \item Draw $X_t\sim P$ from the agent's distribution and compute a conditional e-value $E_t$ from $X_t$ satisfying~\eqref{eq:online-null-validity}.
        \item If $E_t \geq 1/\Delta_t$ and $W_{t-1} + c_tE_t \leq B_t$ then the agent gets allocated $A_t = c_tE_t$. Otherwise set $A_t=0$.
        \item $W_t \gets W_{t-1} + A_t$
    \end{enumerate}
\end{enumerate}
\end{tcolorbox}

The online mechanism draws inspiration from the r-LOND~\citep{javanmard2018online,zrnic2021asynchronous} and e-LOND~\citep{xu2024online} procedures for online multiple testing~\citep{foster2008alpha}. In particular, given a global desired error rate of $\Delta$, it defines a sequence of nominal levels $\Delta_t$, allocating the amount $c_tE_t$ if the e-value exceeds $1/\Delta_t$. As usual, additional care is required since we are not rejecting hypotheses but determining allocations.
The online mechanism gives the following guarantee, the proof of which is deferred to Appendix~\ref{app:proofs}.
For the stopping-time statement, define $\nap_\infty(\calM)=\limsup_{t\to\infty}\nap_t(\calM)$.

\begin{theorem}
\label{thm:online-mechanism}
    The online mechanism satisfies \ref{eq:ic} and the time-uniform guarantee
    \begin{equation}
    \label{eq:anytime-online-control}
        \E_N\left[\sup_{t\geq1}
        \frac{\nap_t(\Monline_\Delta)}{\Delta}\right]\leq1.
    \end{equation}
    Consequently, $\E_N[\nap_\tau(\Monline_\Delta)]\leq\Delta$ for every (potentially infinite) stopping time $\tau$.
\end{theorem}

Theorem~\ref{thm:online-mechanism} provides anytime validity for a fixed, prespecified $\Delta$. It should not be confused with post-hoc risk control over the choice of $\Delta$.

We do not comment on how to set $\gamma_t$ here, except to say that
allowing it  to be predictable turns (one facet of) the problem into a dynamic resource-allocation problem. Under a model for future arrivals, the principal may choose $\gamma_t$ based on the current financial budget, remaining risk budget $1-\sum_{s<t}\gamma_s$, and/or observed history, using dynamic programming or stochastic optimization.

\subsection{Two-stage contracts}
A natural extension of our framework allows agents to use their allocation to fund a second stage of experimentation. Concretely, after receiving allocation $A_i$ based on e-value $E_i^{(1)}$, suppose agent $i$ collects additional data which produces a second e-value $E_i^{(2)}$ conditional on what happened in the first round (the data and allocations), and the combined evidence is the product $E_i^{(1)} \cdot E_i^{(2)}$.
Let $\calF_1 = \sigma(X_1,\dots,X_m)$ be the $\sigma$-algebra generated by all first-round data. For every admissible joint law $H$ and each $i\in N$, assume $\E_H[E_i^{(2)}\mid\calF_1]\leq1$. This conditional validity requirement allows the second-stage sample size to depend on the $\calF_1$-measurable allocation $A_i$. Since $E_i^{(1)}$ is $\calF_1$-measurable, the tower property under $H$ gives
\[
\E_H[E_i^{(1)} \cdot E_i^{(2)}]=\E_H[E_i^{(1)} \cdot \E_H[E_i^{(2)}\mid\calF_1]]\leq\E_H[E_i^{(1)}]\leq1.
\]
Taking the supremum over admissible laws $H$ gives $\E_N[E_i^{(1)} \cdot E_i^{(2)}]\leq1$.
The combined e-values can then be fed into any of our batch mechanisms for a second round of allocation, with risk control guarantees intact, provided the second round is run over the original set of agents (all $m$ of them) with total cost $C$.

Agents who do not continue to the second stage can be assigned $E_i^{(2)}=0$, which trivially satisfies the conditional validity requirement. (Running the second round only over the agents funded in the first round can instead invalidate the guarantees. More specifically, this set of agents would be selected using the first-round e-values, which are therefore biased upwards.) Alternatively, a principal who wishes to restrict the second round to funded agents can instead use the e-values $E_i^{(2)}$ alone, which remain conditionally valid.

\section{Conclusion}

We have studied the problem of resource allocation between a principal and multiple competing agents, each with private information. Our framework is based on \emph{statistical contracts}: protocols which govern the interaction between the principal and the agents. We have focused on questions of how to design such contracts to achieve various objectives, such as disincentivizing ``ineffective'' agents from entering the contract, and bounding the fraction paid out to such agents.

Overall, we believe this to be a fruitful lens through which to view the problem of strategic resource allocation, blending features of statistical inference, multiple testing, game theory, and contract theory.

We hope to have shown that statistical contract theory, while in its infancy, can provide guiding principles on how to construct useful coordination mechanisms between agents with different information and incentives. Broadly speaking, we see this line of research as providing solutions to emerging problems at the intersection of statistics and economics, with data and data analysis playing the role of shared resources and trading opportunities in networked collectives.

\subsection*{Acknowledgements}
BC and AR acknowledge support from NSF grants IIS-2229881 and DMS-2310718. MJ acknowledges support from the Chair ``Markets and Learning,'' supported by Air Liquide, BNP PARIBAS ASSET MANAGEMENT Europe, EDF, Orange and SNCF, sponsors of the Inria Foundation, and from the European Union (ERC-2022-SYG-OCEAN-101071601).
Views and opinions expressed are however those of the author(s) only and do not
necessarily reflect those of the European Union or the European Research Council
Executive Agency. Neither the European Union nor the granting authority can be
held responsible for them.

{
\iflong\small\fi
\bibliographystyle{plainnat}
\bibliography{main}
}

\appendix

\section{Omitted Proofs}
\label{app:proofs}

\begin{proof}[Proof of Proposition~\ref{prop:polytope-separation}]
Fix a candidate allocation $\bar A$ that satisfies the coordinate-wise and budget constraints (i.e., lies in $\calX(\calE^\w)$) and write $\bar W=\|\bar A\|_1$, $F_i=c_iE_i$, and $q=\Delta\bar W$. For $S\subseteq[m]$, write $\bar A(S)=\sum_{i\in S}\bar A_i$ and $F(S)=\sum_{i\in S}F_i$, and define its constraint violation by
\[
    \Phi(S)=C_S\bar A(S)-qF(S).
\]

The subset constraints hold if and only if $\max_S\Phi(S)\leq0$, and including the empty set does not change the separation problem because $\Phi(\emptyset)=0$. With $x_i=\ind{i\in S}$, expansion of the product $C_S\bar A(S)$ gives
\begin{equation}
\label{eq:separation-quadratic}
    \Phi(x)
    =\sum_{i=1}^m(c_i\bar A_i-qF_i)x_i
    +\sum_{i<j}(c_i\bar A_j+c_j\bar A_i)x_ix_j.
\end{equation}
All the pairwise coefficients $c_i \bar A_j + c_j \bar A_i$ are nonnegative, so $\Phi$ is what's known as a
quadratic supermodular set function. Equivalently, $-\Phi$ is quadratic submodular and can be minimized by an $s$--$t$ minimum cut. This is a reasonably well-known optimization fact (see, e.g., \citet{hammer1965some,kolmogorov2002energy}), but let us give an overview of the reduction for completeness.  Set
\[
    d_i=qF_i-c_i\bar A_i,
    \qquad
    w_{ij}=c_i\bar A_j+c_j\bar A_i\geq0.
\]
Then write
\[
    -\Phi(x)=\sum_i d_ix_i-\sum_{i<j}w_{ij}x_ix_j=\sum_i\left(d_i-\frac12\sum_{j\neq i}w_{ij}\right)x_i+\frac12\sum_{i<j}w_{ij}|x_i-x_j|.
\]
We claim that $\Phi$ represents the capacities of a particular $s$--$t$ graph (for more background on $s$--$t$ graphs see \citet[Chapter~6]{ahuja1993network}).
To construct the graph, let $h_i=d_i-\frac12\sum_{j\neq i}w_{ij}$, take nodes $\{s,t,1,\dots,m\}$, and add arcs $i\to j$ and $j\to i$, each of capacity $w_{ij}/2$, for every $i<j$. Add arcs $s\to i$ of capacity $\max\{-h_i,0\}$ and $i\to t$ of capacity $\max\{h_i,0\}$. If $x_i=1$ indicates that node $i$ lies on the source side, the cut capacity is $-\Phi(x)+\sum_i\max\{-h_i,0\}$. A minimum cut therefore returns a set $S^*$ maximizing $\Phi(S)$. If $\Phi(S^*)>0$, the oracle returns the violated valid inequality
\[
    C_{S^*}\sum_{i\in S^*}A_i
    -\Delta F(S^*)\sum_{i=1}^m A_i\leq0.
\]
Otherwise all subset constraints hold. Checking the coordinate-wise and budget constraints takes $O(m)$ operations. The graph has $m+2$ vertices and $O(m^2)$ edges and can be constructed in $O(m^2)$ operations. The so-called ``relabel-to-front maximum-flow'' algorithm~\citep{goldberg1988new} computes a minimum cut in $O(m^3)$ arithmetic operations. For rational input data, the capacities and separating inequalities have polynomial bit length. The final claim follows from the polynomial-time equivalence of separation and linear optimization over rational polytopes~\citep{grotschel1988geometric}.
\end{proof}

\begin{proof}[Proof of Proposition~\ref{prop:compact-safe-allocation-polytope}]
Summing the constraint $CA_i \leq \Delta c_i E_i \|A\|_1$ over $i\in S$ gives \[C\sum_{i\in S}A_i\leq\Delta \|A\|_1 \sum_{i\in S}c_iE_i.\] Since $C_S\leq C$, this implies the corresponding subset constraint in the definition of $\P_\Delta$, thus proving the inclusion. Taking $S=N$ gives
\begin{equation}
\label{eq:compact-polytope-risk-bound}
    \frac{\nap_N(A)}{\Delta}
    \leq \frac{1}{C}\sum_{i\in N}c_i E_i,
\end{equation}
whenever $\|A\|_1>0$ (when $\|A\|_1=0$, the NAP is zero). Taking the supremum over $\Delta$ and then expectations proves the post-hoc risk bound. Null deterrence follows from $A_i\leq c_iE_i$ as in Theorem~\ref{thm:continuous-closure}.
\end{proof}

\begin{proof}[Proof of Proposition~\ref{prop:mechanisms-in-polytope}]
For $\Mcum_\Delta$, activation gives $\|\tilde A\|_1=C/\Delta$ and $0\leq\tilde A_i\leq F_i$. Consequently,
\[
    C\tilde A_i\leq CF_i
    =\Delta F_i\|\tilde A\|_1
    \qquad\text{for every }i,
\]
so $\tilde A\in\Q_\Delta(E)$ before imposing a finite budget.

For $\Mall_\Delta$, a nonzero raw allocation is made only when its total allocation $W$ satisfies $W\geq C/\Delta$. For every funded agent, $\tilde A_i=F_i$, and hence
\[
    C\tilde A_i=CF_i\leq\Delta F_iW.
\]
The inequality is immediate for unfunded agents. Thus its raw allocation satisfies the defining safety constraints of $\Q_\Delta(E)$.

Finally, for $\Mbh_\Delta$, if $k^*>0$, the total raw allocation satisfies
\[
    W=\sum_{i=1}^{k^*}F_{(i)}
    \geq k^*F_{(k^*)}\geq C/\Delta,
\]
and the same argument applies. The zero allocation also belongs to $\Q_\Delta(E)$, and common-factor budget rescaling preserves all of its constraints. Proposition~\ref{prop:compact-safe-allocation-polytope} completes the inclusion claim.
\end{proof}

\begin{proof}[Proof of Proposition~\ref{prop:implicit-objectives}]
Let us deal with each mechanism in turn.

\emph{All-or-nothing.} If~\eqref{eq:polytope-global-threshold} fails, then $\P_\Delta(E)=\Q_\Delta(E)=\{0\}$ by Propositions~\ref{prop:polytope-threshold} and~\ref{prop:compact-safe-allocation-polytope}, and $\Mall$ allocates zero, so there is nothing to prove. Otherwise $F\neq0$, and under budget rescaling the allocation is $\hat A=t^*F$ with $t^*=\min\{1,B/\|F\|_1\}$, which lies in $\Q_\Delta(E)\subseteq\P_\Delta(E)$ by Proposition~\ref{prop:mechanisms-in-polytope}, and $V_{\min}(\hat A)=t^*$. Let $A\in\P_\Delta(E)$ and $v=V_{\min}(A)$. The coordinate caps give $v\leq1$, and also $A_i=0$ whenever $F_i=0$. Since $A_i\geq vF_i$ for every $i$, $B\geq \|A\|_1\geq v\|F\|_1$, so $v\leq B/\|F\|_1$. Hence $v\leq t^*$, and $\hat A$ is a maximizer. For uniqueness, suppose $v=t^*$. If $t^*=1$, then $F\leq A\leq F$. If $t^*=B/\|F\|_1<1$, then $B\geq \|A\|_1\geq t^*\|F\|_1=B$, so every inequality $A_i\geq t^*F_i$ holds with equality. In both cases $A=\hat A$. The same argument applies verbatim over $\Q_\Delta(E)$. When $B=\infty$, $\hat A=F$ and every feasible $A$ satisfies $A\leq F$ coordinate-wise, which proves the final claim.

\emph{Cumulative evidence.} If $\Mcum_\Delta$ does not activate, then~\eqref{eq:polytope-global-threshold} fails and $\Q_\Delta(E)=\{0\}$. Otherwise, let $A\in\Q_\Delta(E)$ be nonzero with support $T$. Summing the constraints $CA_i\leq\Delta F_i\|A\|_1$ over $i\in T$ gives $C\,\|A\|_1\leq\Delta F(T)\|A\|_1$, so $F(T)\geq C/\Delta$. Since the largest $k$ potential allocations have the largest sum among all sets of size $k$, $|T|\geq k^{\rm cum}$. The raw allocation of $\Mcum$ has support of size exactly $k^{\rm cum}$, proportional rescaling multiplies it by a positive constant, and the result lies in $\Q_\Delta(E)$ by Proposition~\ref{prop:mechanisms-in-polytope}.

The activation condition is exactly $\sum_iF_i\geq C/\Delta$, which is also the all-or-nothing condition. Finally, if e-BH activates with index $k^*$, then
\[
    \sum_{j=1}^{k^*}F_{(j)}
    \geq k^*F_{(k^*)}\geq C/\Delta,
\]
so $k^{\rm cum}\leq k^*$. For example, with two agents whose potential allocations are $60$ and $40$ and with $C/\Delta=100$, the cumulative mechanism activates but e-BH does not.

\emph{e-BH.} If $k^*=0$, there is no nonempty self-consistent set: such a $T$ with $|T|=k$ would give $kF_{(k)}\geq k\min_{i\in T}F_i\geq C/\Delta$. Otherwise, let $T^*$ index $F_{(1)},\dots,F_{(k^*)}$. It is self-consistent because $F_{(k^*)}\geq C/(\Delta k^*)$. If $k^*=m$, then $T^*=[m]$ and the claims are immediate. Hence suppose $k^*<m$. Maximality of $k^*$ gives $(k^*+1)F_{(k^*+1)}<C/\Delta$, so $F_{(k^*+1)}<C/(\Delta k^*)$. Now let $T$ be self-consistent with $|T|=k$. As above, $kF_{(k)}\geq C/\Delta$, so $k\leq k^*$, and every $i\in T$ has $F_i\geq C/(\Delta k)\geq C/(\Delta k^*)>F_{(k^*+1)}$, so $i\in T^*$. Hence $F\mathbf{1}_T\leq F\mathbf{1}_{T^*}$ coordinate-wise, which proves the final claim.
\end{proof}

\begin{proof}[Proof of Theorem~\ref{thm:capped-participation}]
All convergence statements in this proof are under $H$. Write $\bar F_{j,n}=c_j\bar E_{j,n}$ and $W_n=\sum_j\tilde A_{j,n}$. Fix an agent $i\in S$ satisfying~\eqref{eq:cap-reaching-power}, and define
\[
    G_{i,n}=\{E_{i,n}\geq L_n\},\text{ and }Z_n=\frac{1}{L_n}\sum_{j\in N}c_jE_{j,n}.
\]
Null validity gives $\E_H[Z_n]\leq C_N/L_n\to0$, so $Z_n\to0$ in probability. For all sufficiently large $n$, $c_iL_n\geq C/\Delta$. On $G_{i,n}$, agent $i$ then alone triggers AoN. For e-BH, if this agent has rank $r$, then $r\bar F_{(r),n}=rc_iL_n\geq C/\Delta$, so $k^*\geq r$ and the agent is selected. Thus, under either policy, $\tilde A_{i,n}=c_iL_n$ on $G_{i,n}$. Regardless of the other non-null agents' e-values,
\[
    W_n\leq\sum_j\bar F_{j,n}
    \leq C_SL_n+\sum_{j\in N}c_jE_{j,n}
    =L_n(C_S+Z_n),
\]
and budget rescaling gives
\[
    A_{i,n}\geq\ind{G_{i,n}}
        c_i\min\{L_n,B/(C_S+Z_n)\}.
\]
The right-hand side is bounded by $Bc_i/C_S$ and converges in probability to that constant, since $\mathbb{P}_H(G_{i,n})\to1$. Its expectation therefore converges to $Bc_i/C_S$, proving $\liminf_n\E_H[A_{i,n}]\geq Bc_i/C_S$. Subtracting $c_i$ proves~\eqref{eq:capped-participation-lower-bound} (note that there are no independence assumptions).

Now suppose~\eqref{eq:cap-reaching-power} holds for every $i\in S$. Then $\bar E_{i,n}/L_n\to1$ for $i\in S$, while $\bar F_{i,n}/L_n\leq Z_n\to0$ for $i\in N$. Thus
\begin{equation}
\label{eq:capped-potential-limit}
    \frac{\bar F_{i,n}}{L_n}
    \longrightarrow c_i\ind{i\in S}
    \text{ in probability.}
\end{equation}

By the preceding selection argument, all non-null agents are selected with probability tending to one. Any selected null agent's raw allocation is bounded by $\bar F_{i,n}$. Consequently, for either mechanism,
\[
    \frac{\tilde A_{i,n}}{L_n}\longrightarrow c_i\ind{i\in S},\text{ and }\frac{W_n}{L_n}\longrightarrow C_S,
\]
in probability. In particular, $W_n\to\infty$ in probability. The budget binds with probability tending to one and $A_{i,n}=B\tilde A_{i,n}/W_n$ on that event. Thus $A_{i,n}$ converges in probability to $Bc_i/C_S$ for $i\in S$ and to zero for $i\in N$. Since $0\leq A_{i,n}\leq B$, the expectations converge to the same limits.

The total allocation converges in probability to $B>0$, while the allocation to null agents converges to zero. Thus $\nap_N(A_n)\to0$ in probability, and its boundedness by one gives $\E_H[\nap_N(A_n)]\to0$.
\end{proof}

\begin{proof}[Proof of Proposition~\ref{prop:finite-sample-participation}]
Write $\bar F_j=c_j\min\{E_j,L\}$. On $\{E_i\geq L\}$, agent $i$ has potential allocation $c_iL\geq C/\Delta$, which alone triggers AoN. For e-BH, let $r$ be agent $i$'s rank in any ordering used to break ties. Then $\bar F_{(r)}=c_iL$, so $r\bar F_{(r)}\geq C/\Delta$. Hence $k^*\geq r$ and agent $i$ is selected. Under either policy, the raw allocation to this agent is $c_iL$ and the total raw allocation is at most $CL$. Proportional budget rescaling therefore gives
\[
    A_i\geq c_iL\min\{1,B/(CL)\}=c_i\min\{L,B/C\}.
\]
Allocations are nonnegative off this event. Taking expectations proves the payment bound, and comparing it to $c_i$ gives~\eqref{eq:finite-sample-participation}.
\end{proof}

\begin{proof}[Proof of Proposition~\ref{prop:payment-floors}]
The truncated baseline belongs to $\Q_\Delta(\bar E)\subseteq\Q_\Delta(E)\subseteq\P_\Delta(E)$, where $\bar E$ denotes the truncated evidence. Since $b_i(E)\geq\lambda b_i(E)$, it is feasible for the payment-floor problem. The additional constraints preserve membership in $\P_\Delta(E)$, so \ref{eq:ic} and post-hoc risk control follow from Theorem~\ref{thm:continuous-closure}.

The region $\P_\Delta(E)$ is compact and convex. The constrained region is also compact, and continuity of $V$ ensures that both maxima are attained. By convexity, the allocation
\[
    A^{\rm mix}=(1-\lambda)A^*+\lambda b(E)
\]
belongs to the chosen region and satisfies the payment floors because $A^*\geq0$. Optimality and concavity therefore give
\[
    V(A^\lambda;E)\geq V(A^{\rm mix};E)\geq(1-\lambda)V(A^*;E)+\lambda V(b(E);E).
\]
Rearranging proves~\eqref{eq:payment-floor-objective-bound}.
Finally, the bound~\eqref{eq:payment-floor-participation} follows by taking expectations in $A^\lambda_{i,n}\geq\lambda b_{i,n}$ and applying the individual lower bound of Theorem~\ref{thm:capped-participation}.
\end{proof}

\begin{proof}[Proof of Theorem~\ref{thm:online-mechanism}]
For null deterrence, each agent is allocated at most $c_iE_i$, hence~\eqref{eq:online-null-validity} and the tower property give $\E[A_i\ind{i\in N}]\leq c_i\E[\mathbb{P}(i\in N\mid\calF_{i-1})]=c_i\mathbb{P}(i\in N)$.

Next, for all $t$, $\sum_{j\in[t]} A_j = W_t$.
    For each $i$, let
    \[I_i = \ind{E_i \geq 1/\Delta_i}\ind{W_{i-1} + c_iE_i\leq B_i}\] be the indicator that agent $i$ satisfies the allocation rule.
    Write
    \[
        \nap_t(\Monline_\Delta)=\frac{\sum_{i\in N\cap[t]} A_i}{W_t}=\frac{\sum_{i\in N\cap[t]} A_i I_i}{W_t}\leq \sum_{i\in N\cap[t]}\frac{A_i I_i}{W_{i-1}+A_i},
    \]
    where the inequality follows from the definition of $W_t$ as the total allocation at time $t$.
    Now suppose $I_i=1$, so that $A_i = c_iE_i \geq c_i/\Delta_i$. This forces $\gamma_i>0$: if $\gamma_i=0$, then $\beta_i=1$ and $\Delta_i=0$, so $1/\Delta_i = \infty > E_i$. If $\beta_i>0$, then $\Delta_i=\Delta\gamma_i/\beta_i$ and
    \[W_{i-1} + A_i \geq W_{i-1} + c_i/\Delta_i = W_{i-1} + \frac{c_i\beta_i}{\Delta \gamma_i} = \frac{c_i}{\Delta\gamma_i}, \]
    by definition of $\beta_i$. If $\beta_i\leq0$, then $W_{i-1}\geq c_i/(\Delta\gamma_i)$, again by definition of $\beta_i$. In both cases, $W_{i-1}+A_i\geq c_i/(\Delta\gamma_i)>0$ whenever $I_i=1$. Therefore,
    \[
        \sum_{i\in N\cap[t]}\frac{A_i I_i}{W_{i-1}+A_i}\leq \Delta \sum_{i\in N\cap[t]}\gamma_i \frac{A_i}{c_i} I_i=\Delta \sum_{i\in N\cap[t]} \gamma_i E_i I_i,
    \]
    hence, simultaneously for every $t$,
    \begin{equation*}
        \frac{\nap_t(\Monline_\Delta)}{\Delta}
        \leq\sum_{i\in N\cap[t]}\gamma_iE_iI_i
        \leq\sum_{i\in N}\gamma_iE_iI_i.
    \end{equation*}
    Fix any joint law satisfying~\eqref{eq:online-null-validity} and write $\E$ for expectation under that law. Writing $\sum_{i\in N}\gamma_iE_iI_i=\sum_{i=1}^\infty\ind{i\in N}\gamma_iE_iI_i$, taking the supremum over $t$ and then expectations, Tonelli's theorem, predictability, and~\eqref{eq:online-null-validity} give
    \begin{align*}
        \E\left[\sup_{t\geq1}
        \frac{\nap_t(\Monline_\Delta)}{\Delta}\right]
        &\leq \sum_{i=1}^\infty\E[\ind{i\in N}\gamma_iE_iI_i]\leq \sum_{i=1}^\infty\E[\ind{i\in N}\gamma_iE_i]\\
        &=\sum_{i=1}^\infty\E\!\left[
            \gamma_i\E[E_i\ind{i\in N}\mid\calF_{i-1}]\right]\leq\sum_{i=1}^\infty\E[\gamma_i]
        =\E\!\left[\sum_{i=1}^\infty\gamma_i\right]
        \leq1.
    \end{align*}
    This proves~\eqref{eq:anytime-online-control}. For every potentially infinite stopping time $\tau$, $\nap_\tau/\Delta$ is bounded pathwise by the same supremum, using the definition of $\nap_\infty$ on $\{\tau=\infty\}$, proving the stopping-time claim.
\end{proof}

\section{Additional Results}

\subsection{Alternative scaling method for finite budgets}
\label{sec:finite-budget-methods}

The explicit mechanisms in Section~\ref{sec:nar-control} use a rescaling method which spends the entire budget whenever the total raw allocation exceeds $B$. Let us call this \emph{hard-cap} rescaling. For a finite budget $0<B<\infty$, an alternative is \emph{evidence-adaptive reserve rescaling}. Given raw allocations $\tilde A_i$ with total $W=\sum_i\tilde A_i$, set
\begin{equation}
\label{m2}
    A_i=\frac{B\tilde A_i}{B+W}.
\end{equation}
This can be interpreted as allocating the budget between the real agents $1,\dots,m$ and a dummy reserve agent with a maximum potential allocation of $B$. The total allocation to real agents and the unallocated reserve are, respectively,
\begin{equation}
\label{eq:evidence-adaptive-reserve}
    T(W)=\frac{BW}{B+W},\text{ and }R(W)=B-T(W)=\frac{B^2}{B+W}.
\end{equation}
Thus the fraction of the budget held in reserve decreases as the total raw allocation grows. Since $A_i\leq\tilde A_i$ and all coordinates are multiplied by the same positive factor, the reserve rule preserves null deterrence and leaves NAP unchanged. In particular, the three explicit policies retain their post-hoc NAP risk guarantees.

The reserve rule may be useful when false allocations are especially costly. To see the tradeoff, suppose the principal's utility is linear, $U(A_1^m)=\sum_i u(A_i,P_i)$, with nonpositive contributions from null agents and nonnegative contributions from non-null agents.
(This is the principal's ``true'' utility, distinct from the objective $V$ it is maximizing in Sections~\ref{sec:general-losses} and~\ref{sec:nap}.)
For $W>0$, the hard-cap and reserve rules multiply the raw utility by $\min\{1,B/W\}$ and $B/(B+W)$, respectively. Reserve rescaling therefore gives higher utility precisely when $U(\tilde A_1^m)<0$, and the hard-cap rule gives higher utility when it is positive. If $W=0$, both allocate zero. For coordinate-wise nondecreasing objectives, the hard-cap rule is always weakly preferable because it pays at least as much to every agent.

\subsection{Extension to approximate and asymptotic e-values}
\label{app:approximate-evalues}

E-values are known for a wide variety of settings; see the book by \citet{ramdas2025hypothesis} for an overview. Still, just as in the world of p-values, for some classes of distributions no  nontrivial e-value is known, and agents may instead rely on e-values that are valid only approximately or asymptotically~\citep{ignatiadis2024asymptotic,chugg2026posthoc,martinez2026bentkus}.
Here we investigate the extent to which our contracts can be made to handle such objects.
Overall, we find that many of the theorems degrade gracefully with the approximation error of the e-value, provided one uses the right notion of approximation.

Following \citet[Definition~3.1]{ignatiadis2024asymptotic}, for $\varepsilon\geq0$ and $\delta\in[0,1]$, we call a nonnegative random variable $E$ an \emph{$(\varepsilon,\delta)$-approximate e-value} for $\calP$ if
\begin{equation}
\label{eq:approximate-evalue}
    \E_P[E\wedge t]\leq1+\varepsilon+\delta t\text{ for all }t\geq0\text{ and }P\in\calP,
\end{equation}
where $a\wedge b=\min\{a,b\}$. Letting $t\to\infty$ shows that $E$ is $(\varepsilon,0)$-approximate if and only if $\E_P[E]\leq1+\varepsilon$ for all $P\in\calP$, i.e., if and only if $E/(1+\varepsilon)$ is an e-value. A positive $\delta$ instead allows the bound on the expectation to fail on an event of probability at most $\delta$~\citep[Proposition~3.2]{ignatiadis2024asymptotic}. A sequence $(E_n)$ is a \emph{uniformly asymptotic} e-value if each $E_n$ is $(\varepsilon_n,\delta_n)$-approximate with $\varepsilon_n,\delta_n\to0$, and it is \emph{uniformly strongly asymptotic} if, in addition, $\delta_n=0$ for all large $n$. The latter is equivalent to $\limsup_n\sup_{P\in\calP}\E_P[E_n]\leq1$~\citep[Proposition~3.6]{ignatiadis2024asymptotic}. \citet{ignatiadis2024asymptotic} also allow $\varepsilon$ and $\delta$ to depend on $P$, in which case convergence in distribution to an e-value under each $P\in\calP$ suffices for asymptotic validity (their Proposition~3.8).

Let us turn now to using these objects in our mechanisms.
Suppose each agent $i$ uses an $(\varepsilon_i,\delta_i)$-approximate e-value $E_i$, and write
\[
    \bar\varepsilon_N=\frac{1}{C_N}\sum_{i\in N}c_i\varepsilon_i,\text{ and }\delta_N=\sum_{i\in N}\delta_i.
\]
If the mechanisms are run on $E=(E_1,\dots,E_m)$ exactly as before, then we obtain the following guarantees. The upshot is that when using approximate e-values, null deterrence and risk bounds also hold approximately.

\begin{proposition}
\label{prop:approximate-evalues}
Let $B\in(0,\infty]$.
\begin{enumerate}[nosep,topsep=0.1cm]
    \item Suppose $\delta_i=0$ for every $i\in N$. Every measurable selection $A\in\calX(\calE^\w)$ satisfies $\E_N[A_i]\leq c_i(1+\varepsilon_i)$ for each $i\in N$. For any family of losses $\calL$, the left-hand side of~\eqref{eq:continuous-closure-expectation} with $\calE=\calE^\w$ is at most $1+\bar\varepsilon_N$. In particular, mechanisms selecting from $\P_\Delta(E)$ for every $\Delta>0$ satisfy $\risk_N(\calM,\nap)\leq1+\bar\varepsilon_N$, and those selecting from $\Q_\Delta(E)$, including $\Mcum$, $\Mall$, and $\Mbh$, satisfy $\risk_N(\calM,\nap)\leq C_N(1+\bar\varepsilon_N)/C$.
    \item In general, for every $\Delta>0$, every measurable selection $A\in\P_\Delta(E)$ satisfies $\E_N[\nap_N(A)]\leq\Delta(1+\bar\varepsilon_N)+\delta_N$. If $B<\infty$, every measurable selection $A\in\calX(\calE^\w)$ satisfies $\E_N[A_i]\leq c_i(1+\varepsilon_i)+\delta_iB$ for each $i\in N$.
\end{enumerate}
\end{proposition}

\begin{proof}
Fix a joint law under which $N$ is the set of null agents, and write $\E$ for its expectation. For each $i\in N$, the marginal law of agent $i$ lies in $\calP$, so $\E[E_i\wedge t]\leq1+\varepsilon_i+\delta_it$ for every $t\geq0$. In part~1, letting $t\to\infty$ and applying monotone convergence gives $\E[E_i]\leq1+\varepsilon_i$. Each bound in part~1 then follows from a pathwise inequality established in the main text, in which the relevant quantity is dominated by a nonnegative linear combination of the null agents' e-values: $A_i\leq c_iE_i$ for null deterrence, the bound by $E_N^\w=C_N^{-1}\sum_{i\in N}c_iE_i$ in the proof of Theorem~\ref{thm:continuous-closure}, and~\eqref{eq:compact-polytope-risk-bound} for $\Q_\Delta(E)$. Taking expectations and using linearity gives the claims. No assumption on the dependence between agents is needed.

For part~2, let $A\in\P_\Delta(E)$. If $\|A\|_1=0$ then $\nap_N(A)=0$. Otherwise, the constraint~\eqref{eq:pointwise-nap-certificate} with $S=N$ gives $\nap_N(A)\leq\Delta C_N^{-1}\sum_{i\in N}c_iE_i$. Since also $\nap_N(A)\leq1$ and $x\mapsto1\wedge x$ is subadditive on $[0,\infty)$,
\[
    \nap_N(A)\leq1\wedge\sum_{i\in N}\frac{\Delta c_iE_i}{C_N}
    \leq\sum_{i\in N}\left(1\wedge\frac{\Delta c_iE_i}{C_N}\right)
    =\sum_{i\in N}\frac{\Delta c_i}{C_N}\left(E_i\wedge\frac{C_N}{\Delta c_i}\right).
\]
Applying~\eqref{eq:approximate-evalue} with $t=C_N/(\Delta c_i)$ to each term gives $\E[\nap_N(A)]\leq\sum_{i\in N}\left(\Delta c_i(1+\varepsilon_i)/C_N+\delta_i\right)=\Delta(1+\bar\varepsilon_N)+\delta_N$. For null deterrence, $A_i\leq c_iE_i$ and $A_i\leq \|A\|_1\leq B$ give $A_i\leq c_i(E_i\wedge B/c_i)$, and we apply~\eqref{eq:approximate-evalue} with $t=B/c_i$. Taking the supremum over joint laws completes the proof.
\end{proof}

Applied to e-values $E_{i,n}$ computed from $n$ observations, part~1 shows that if each agent's e-value is uniformly \emph{strongly} asymptotic, then $\limsup_n\E_N[A_{i,n}]\leq c_i$ for every $i\in N$ and the post-hoc risk is asymptotically at most one (at most $C_N/C$ for $\Q_\Delta(E)$). Part~2 shows that if they are merely uniformly asymptotic, then null deterrence still holds asymptotically when $B<\infty$, and $\limsup_n\E_N[\nap_N(A_n)]\leq\Delta$ for each fixed $\Delta$. But post-hoc risk control may not hold.

The restriction $\delta_i=0$ in part~1 cannot be dropped, as the following example shows.

\begin{example}
\label{ex:t-statistic-evalue}
Let $m=1$, and suppose the agent's null hypothesis is that their iid observations have mean zero and finite, positive variance. Consider
\[
    E_n=\exp\left(\lambda T_n-\frac{\lambda^2}{2}\right),
\]
where $T_n=\sqrt n\bar X_n/\hat\sigma_n$ is the t-statistic, $\hat\sigma_n$ is the sample standard deviation, and $\lambda>0$. Under every null, $T_n$ converges in distribution to $\mathcal N(0,1)$, so $(E_n)$ is an asymptotic e-value~\citep[Section~6.2]{ignatiadis2024asymptotic}. It is not strongly asymptotic, however. If the observations are Gaussian, then $T_n$ has a Student t distribution with $n-1$ degrees of freedom, whose moment generating function is infinite, so $\E_P[E_n]=\infty$ for every $n\geq2$. With a single agent, $\Mcum_\Delta$, $\Mall_\Delta$, and $\Mbh_\Delta$ each allocate if and only if $E_n\geq1/\Delta$, in which case the NAP equals one. Thus, if the agent is null, $\sup_{\Delta>0}\nap_N(\calM_\Delta)/\Delta=E_n$, and the post-hoc risk is infinite for every $n$ and every budget. If $B=\infty$, then $\Mall_\Delta$ also gives the null agent expected allocation $c_1\E_P[E_n\ind{E_n\geq1/\Delta}]=\infty$, so null deterrence fails as well.
\end{example}

The failure in Example~\ref{ex:t-statistic-evalue} is due to the particular construction rather than to the null hypothesis. If $\hat\sigma_n$ is replaced by the uncentered root mean square of the observations, then $(E_n)$ becomes strongly asymptotic~\citep{ignatiadis2024asymptotic,chugg2026posthoc}. The Bentkus-type asymptotic e-values of \citet{martinez2026bentkus}, which apply polynomial rather than exponential transformations to this self-normalized statistic, are also strongly asymptotic and avoid the conservativeness of the exponential construction. Since both are valid pointwise in $P$, the limiting guarantees of part~1 hold for them under each fixed joint law.

Truncation, as in Section~\ref{sec:nonnull-participation}, converts asymptotic e-values into strongly asymptotic ones. Indeed, if $E$ is $(\varepsilon,\delta)$-approximate, then $\E_P[E\wedge L]\leq1+\varepsilon+\delta L$, so $E\wedge L$ is $(\varepsilon+\delta L,0)$-approximate. Hence if $(E_{i,n})$ is uniformly asymptotic, then $(E_{i,n}\wedge L_n)$ is uniformly strongly asymptotic whenever $\delta_nL_n\to0$, and in particular for any fixed truncation level $L_n=L$. Part~1 of Proposition~\ref{prop:approximate-evalues} then applies to the mechanisms with truncated e-values, including post-hoc risk control.

\subsection{Truncation is necessary for growth-independent allocations}
\label{app:truncation-necessity}

Theorem~\ref{thm:capped-participation} shows that if a non-null agent's evidence reaches the truncation level, its limiting expected allocation is at least $Bc_i/C_S$, regardless of how quickly the other agents' evidence grows. Remark~\ref{rem:log-thresholding} shows that replacing $E_i$ by $\log(1+E_i)$ also prevents the vanishing share in Example~\ref{ex:participation-dilution}, but the resulting allocations depend on those rates. The following result shows that without such assumptions, truncation cannot be replaced by any unbounded transformation of the evidence.

\begin{proposition}
\label{prop:unbounded-transformation-dilution}
Fix $c_1,c_2>0$, $\Delta>0$, $0<B<\infty$, and a nondecreasing unbounded function $f:[0,\infty)\to[0,\infty)$ satisfying $f(x)\leq x$. There exist a null hypothesis, e-value sequences $(E_{1,n})$ and $(E_{2,n})$ valid under that null for every $n$, and a joint law $H$ under which both agents are non-null, such that $f(E_{i,n})\to\infty$ almost surely for $i=1,2$, but running either $\Mall$ or $\Mbh$ on $(f(E_{1,n}),f(E_{2,n}))$ gives $A_{1,n}\to0$ almost surely and
\[
\lim_{n\to\infty}\E_H[A_{1,n}]=0.
\]
\end{proposition}

\begin{proof}
Since $f$ is unbounded, choose deterministic $a_n,b_n\geq1$ such that
\[
    f(a_n)\geq n,\text{ and }f(b_n)\geq n f(a_n).
\]
Let the null hypothesis be that an agent's observations are iid uniform on $[0,1]$, and define
\[
    E_{1,n}=a_n\ind{X_{1,1}\leq1/a_n},\text{ and }
    E_{2,n}=b_n\ind{X_{2,1}\leq1/b_n}.
\]
Each statistic uses only the first of the agent's $n$ observations and has expectation one under the null, so both sequences consist of valid e-values. Since $0\leq f(x)\leq x$, their transformed values are also valid e-values. Take $H$ to be the joint law under which every observation of both agents equals zero. Both agents are then non-null, and $E_{1,n}=a_n$ and $E_{2,n}=b_n$ almost surely for every $n$. In particular, both transformed sequences diverge.

Write $C=c_1+c_2$. For all sufficiently large $n$, both $c_1f(a_n)$ and $c_2f(b_n)$ exceed $C/\Delta$, so both agents are selected by either policy. Their total raw allocation also exceeds $B$. Proportional budget rescaling therefore gives, almost surely,
\[
    A_{1,n}=\frac{Bc_1f(a_n)}{c_1f(a_n)+c_2f(b_n)}
    \leq\frac{Bc_1}{c_1+nc_2}\longrightarrow0.
\]
The same deterministic bound gives $\E_H[A_{1,n}]\to0$.
\end{proof}

In particular, even when $B>c_1+c_2$, diverging transformed evidence alone cannot guarantee profitable entry. The proof uses $f$ only at each fixed $n$, so the conclusion also holds if $f$ is replaced by a sequence of transformations $(f_n)$, each satisfying the same conditions. The truncations $x\mapsto\min\{x,L_n\}$ of Theorem~\ref{thm:capped-participation} escape the result precisely because each is bounded, even though $L_n\to\infty$. Once both agents' e-values exceed $L_n$, their truncated evidence is identical, so the allocation no longer depends on how quickly either e-value grows. This is the sense in which truncation is necessary: any unbounded transformation can be defeated by a competitor whose evidence grows sufficiently quickly, whereas under truncation an agent's limiting allocation is guaranteed by its own power alone.

Proposition~\ref{prop:unbounded-transformation-dilution} applies to $f(x)=\log(1+x)$, but it does not contradict Remark~\ref{rem:log-thresholding}. The remark assumes that $n^{-1}\log E_{i,n}\to r_i\in(0,\infty)$ for every non-null agent $i$, i.e., that all non-null agents' evidence grows at finite exponential rates. The construction above violates this assumption. With $f(x)=\log(1+x)$, it requires $\log(1+E_{2,n})=f(b_n)\geq nf(a_n)\geq n^2$, so $n^{-1}\log E_{2,n}\to\infty$ and agent 2's evidence grows super-exponentially.

The two results are therefore complementary. The logarithm guarantees every non-null agent a nonvanishing share provided all non-null agents' evidence grows at finite, positive exponential rates, but the resulting shares depend on those rates. Proposition~\ref{prop:unbounded-transformation-dilution} shows that some restriction on competitors' growth is unavoidable for any unbounded transformation. Truncation requires no such restriction, and its limiting allocations do not depend on the rates at all.

\subsection{Selective entry and the ex-post wrapper}
\label{app:selective-entry}

Suppose agents may enter the contract after collecting their data. Returning to Example~\ref{ex:pharma}, a biotech startup may naturally prefer to submit evidence from an existing private trial rather than opt into a contract that requires a new experiment. Formally, replace step (2) of the statistical contract in Section~\ref{sec:contract} by
\begin{quote}
(2) Each agent $i$ with private distribution $P_i$ may generate data $X_i\sim P_i$ and, after observing it, decide whether to pay $c_i$ and provide the data to the principal.
\end{quote}
We call this the \emph{ex-post statistical contract}. It can be handled by treating nonparticipation as zero submitted evidence.

\begin{tcolorbox}[
    colback=pink,
    colframe=white,
    boxrule=0pt,
    sharp corners,
    left=6pt,
    right=6pt,
    top=6pt,
    bottom=6pt,
    breakable
]
\textbf{Ex-post wrapper}
\begin{enumerate}[nosep,topsep=0.1cm]
    \item Let $\bar m$ be the total number of agents who may enter. For each $i\in[\bar m]$, set $\tilde E_i=E_iI_i$, where $I_i$ indicates whether agent $i$ enters.
    \item Run continuous-closure, $\Mcum$, $\Mall$, or $\Mbh$ on all $\bar m$ agents using $\tilde E_1,\dots,\tilde E_{\bar m}$. For continuous-closure, use their cost-weighted e-collection and choose a measurable allocation from the resulting region.
\end{enumerate}
\end{tcolorbox}
The variables $\tilde E_i$ are observable: if agent $i$ does not enter, then $I_i=0$ and hence $\tilde E_i=0$. The wrapper requires the principal to know the entire eligible pool in advance.

\begin{theorem}
\label{thm:ex-post}
Consider running continuous-closure for NAP, $\Mcum$, $\Mall$, or $\Mbh$ with the ex-post wrapper. The resulting mechanism $\calM$ satisfies $\risk_N(\calM,\nap)\leq1$. For $\Mcum$, $\Mall$, and $\Mbh$, the sharper bound $\risk_N(\calM,\nap)\leq C_N/C_{[\bar m]}$ holds. More generally, continuous-closure for a loss satisfying $\ell_S(0)=0$ for every $S$ has a nonempty region and satisfies $\risk_N(\calM,\ell)\leq1$. These statements hold for any budget.
\end{theorem}

\begin{proof}
For each null agent $i\in N$, the variable $\tilde E_i=E_iI_i$ satisfies $\E_N[\tilde E_i]\leq\E_N[E_i]\leq1$, since $I_i\in\{0,1\}$ and $E_i$ is an e-value. Thus $\tilde E_1,\dots,\tilde E_{\bar m}$ are valid e-values, and their cost-weighted averages form a valid e-collection. The wrapper therefore reduces to running the corresponding batch mechanism on these e-values over $\bar m$ agents. For general-loss closure, the zero allocation is feasible because $\ell_S(0)=0$. The conclusions follow from Theorem~\ref{thm:continuous-closure} and Propositions~\ref{prop:compact-safe-allocation-polytope} and~\ref{prop:mechanisms-in-polytope}.
\end{proof}

The wrapper can be conservative. Let $C_{[\bar m]}=\sum_{i=1}^{\bar m} c_i$. Allocation under $\Mcum$ and $\Mall$ requires $\sum_{i\in[\bar m]}c_i\tilde E_i\geq C_{[\bar m]}/\Delta$, while $\Mbh$ requires $kF_{(k)}\geq C_{[\bar m]}/\Delta$. Both conditions become harder to satisfy when the eligible pool is much larger than the participating set. Moreover, although the wrapper preserves the post-hoc risk guarantee, it does not generally yield an ex-post contract that satisfies null deterrence, as discussed next.

\subsubsection{Limits of null deterrence under selective entry}
\label{app:selective-entry-ic}

Recall that without selective entry, null deterrence requires
\[
    \E_N[A_i]\leq c_i,
\]
for each null agent $i\in N$ (i.e., $P_i\in\calP$). Since the agent's decision to opt into the contract is independent of the data they generate after opting in, this is equivalent to $\E_N[I_i(A_i-c_i)]\leq0$, where $I_i\in\{0,1\}$ indicates whether the agent opted into the contract. Under selective entry, however, $\E[I_i(A_i-c_i)]\leq0$ is not equivalent to $\E[A_i]\leq c_i$, since $I_i$ is not independent of $A_i$. And while the latter continues to hold when using the ex-post wrapper (since $\tilde E_i\leq E_i$), the former does not.

In fact, the condition $\E[I_i(A_i-c_i)]\leq0$ cannot be met in general unless no allocation to agent $i$ exceeds $c_i$ with positive probability under $\calP$. Indeed, fix an agent $i\in N$ and consider the counterfactual allocation $A_i^\circ$, which is the allocation were agent $i$ to enter the contract, keeping the behavior of all other agents fixed. Now consider the natural entry rule $I_i=\ind{A_i^\circ>c_i}$. This rule can be calculated once the behavior of the other agents is known.

For any $P\in\calP$ such that $i\in N$, $\E_P[I_i(A_i-c_i)]=\E_P[I_i(A_i^\circ-c_i)]=\E_P[A_i^\circ-c_i\mid I_i=1]P(A_i^\circ>c_i)$. Thus, to ensure that $\E_P[I_i(A_i-c_i)]\leq0$, one needs $P(A_i^\circ>c_i)=0$ for all such $P$. That is, the allocation to each null agent $i$ must be at most $c_i$ almost surely. This is an extremely restrictive requirement. For our e-value-based allocations $A_i=c_iE_i$, for instance, it says that the e-value must never exceed one, thereby allowing only trivial choices.

Note, however, that this argument requires that the agent be able to compute $A_i^\circ$ before entering. That is, the behavior of the other agents is known to them. Thus, while the above argument shows that $\E[I_i(A_i-c_i)]\leq0$ is impossible in the ex-post setting in general, it does not preclude the possibility that this condition might be met in more restrictive settings. Indeed, it would be quite natural to assume that the agents do not know each other's actions, and thus could only estimate $A_i^\circ$. In such a scenario, an appropriate version of null deterrence may exist.

\section{Examples of Loss Functions}
\label{app:loss-examples}

The NAP is not the only loss one could study. Table~\ref{tab:loss-functions}
lists several alternatives.
Continuous-closure applies to every loss for which its region is nonempty. For losses that vanish at the zero allocation, this is automatic. Otherwise nonemptiness must be checked. With a finite budget, feasibility is evaluated on the final allocation.

\begin{table}[ht!]
\centering
\small
\caption{Examples of loss functions the principal may consider.}
\label{tab:loss-functions}
\begin{tabular}{p{0.2\linewidth}p{0.38\linewidth}p{0.3\linewidth}}
\toprule
Loss & Definition & Interpretation \\
\midrule
Null allocation proportion
&
$\nap_N(A)=\sum_{i\in N}A_i/\sum_{j\in[m]}A_j$
&
Share of total allocation going to null agents. \\
\midrule
Total null allocation
&
$\ell_N(A)=\sum_{i\in N}A_i$
&
Absolute amount allocated to null agents. \\
\midrule
Budget-normalized null allocation
&
$\ell_N(A)=B^{-1}\sum_{i\in N}A_i$
&
Fraction of a fixed budget allocated to null agents. \\
\midrule
Weighted null allocation proportion
&
$\ell_N(A)=\sum_{i\in N}w_iA_i/\sum_{j\in[m]}w_jA_j$
&
Weighted share of allocation going to null agents. \\
\midrule
Weighted null allocation
&
$\ell_N(A)=\sum_{i\in N}w_iA_i$
&
Weighted absolute allocation to null agents. \\
\midrule
Maximum null allocation
&
$\ell_N(A)=\max_{i\in N}A_i$
&
Largest allocation made to any null agent. \\
\midrule
Excess null allocation
&
$\ell_N(A)=\left(\sum_{i\in N}A_i-\tau\right)_+$
&
Amount by which null allocation exceeds a tolerance. \\
\midrule
False discovery proportion
&
$\ell_N(A)=|N\cap R(A)|/|R(A)|$, where $R(A)=\{i:A_i>0\}$
&
Unweighted fraction of selected agents that are null. \\
\midrule
False discovery count
&
$\ell_N(A)=|N\cap R(A)|$
&
Number of null agents receiving positive allocation. \\
\midrule
Missed non-null allocation
&
$\ell_N(A)=\sum_{i\notin N}(A_i^\star-A_i)_+$
&
Shortfall relative to desired non-null allocations. \\
\midrule
False non-allocation proportion
&
$\ell_N(A)=|\{i\notin N:A_i=0\}|/|N^c|$
&
Fraction of non-null agents receiving no allocation. \\
\bottomrule
\end{tabular}
\end{table}

\section{Experimental Details}
\label{app:experiments}

We generate independent instances as follows. Each of the \(m\) agents is null with probability \(1/2\). Agent \(i\) supplies \(n=30\) independent observations from \(\mathcal N(0,\sigma^2)\) under the null or \(\mathcal N(1,\sigma^2)\) under the alternative, with \(\sigma^2=4\). The corresponding normal likelihood-ratio e-value is
\[
E_i=\exp\!\left(\frac{1}{\sigma^2}\sum_{t=1}^{n}X_{i,t}
-\frac{n}{2\sigma^2}\right).
\]
Heterogeneous costs and observable weights are drawn independently according to
\[
\log c_i\sim \mathcal N(-0.55^2/2,0.55^2),\text{ and }\log v_i\sim \mathcal N(-\sigma_v^2/2,\sigma_v^2),
\]
so both distributions have mean one. Unless stated otherwise, \(\sigma_v=0.8\). For Figure~\ref{fig:optimality-gap-runtime}, we set \(m=30\), \(B=30\), and \(\Delta=0.1\), and retain 150 independently generated instances on which the full safe polytope admits a nonzero allocation. Conditioning on activation isolates allocation quality from the separate question of whether the evidence suffices to allocate at all. We run this comparison at \(\sigma_v\in\{0,0.2,0.4,\dots,1.2\}\) on the same 150 instances. Writing \(\log v_i=\sigma_vZ_i-\sigma_v^2/2\) for a standard normal \(Z_i\) drawn once per agent, only the dispersion of the weights changes across values of \(\sigma_v\). We also record the realized null allocation proportion of each allocation.

For the linear comparison, we maximize \(V_{\rm lin}(A)=\sum_i v_i A_i\). For the concave comparison, we directly maximize
\(V_{\log}(A)=\sum_i v_i\log(1+A_i)\). Optimization over \(\P_\Delta(E)\) uses constraint generation initialized with the budget constraint. At each iteration, we solve the current optimization problem and use the minimum-cut oracle in Proposition~\ref{prop:polytope-separation} to add a maximally violated subset constraint, stopping when no violation exceeds the numerical tolerance. Optimization over \(\Q_\Delta(E)\) uses all its constraints in a single optimization problem. The linear objective gives a linear program, while the logarithmic objective gives a smooth convex optimization problem. The cumulative-evidence, AoN, and e-BH allocations use hard-cap proportional rescaling. All methods are evaluated using the same objective. We report the relative gap
\[
\frac{V_{\P}^{\star}-V(A)}
     {\max\{V_{\P}^{\star},10^{-12}\}},
\]
where \(V_{\P}^{\star}\) is the objective value from full-polytope optimization for the same realized instance.

For Figure~\ref{fig:optimality-gap-scaling}, we repeat both utility comparisons at \(m\in\{10,20,40,80,100\}\), with 150 independently generated active instances at each size (750 in total). All other parameters, including \(B=30\) and \(\Delta=0.1\), are held fixed. The four methods are compared with full-polytope optimization on the same instances.

The Python implementation uses SciPy's HiGHS interface for linear programs, its SLSQP method with analytic gradients for logarithmic utility, and the NetworkX minimum cut. For each subproblem, we check feasibility and bound the objective gap by maximizing \(\nabla V_{\log}(A)^\top(A'-A)\) over its feasible region, requiring this bound to be at most \(10^{-6}\max\{1,V_{\log}(A)\}\).

\end{document}